\documentclass{lmcs}
\pdfoutput=1

\usepackage{lastpage}
\lmcsdoi{17}{3}{8}
\lmcsheading{}{\pageref{LastPage}}{}{}%
{Jan.~04,~2021}{Jul.~21,~2021}{}

\usepackage[utf8]{inputenc}

\usepackage{amsmath,amssymb}
\usepackage{mathtools}
\usepackage{thm-restate}
\usepackage{xfrac}
\usepackage{stmaryrd,wasysym}
\usepackage{tikz}
\usetikzlibrary{automata,arrows,calc,decorations.pathreplacing,petri,shapes.geometric}
\usepackage{todonotes}
\usepackage[inline]{enumitem}
\usepackage[labelformat=parens]{subcaption}
\usepackage{wrapfig}
\usepackage[misc]{ifsym}

\setlist[description]{topsep=1.7ex}

\input{macros.tex}


\begin{document}

\title[The Inconsistent Labelling Problem]{A Detailed Account of The 
Inconsistent Labelling Problem of Stutter-Preserving Partial-Order 
Reduction\rsuper*}
\titlecomment{{\lsuper*}An extended abstract of this paper appeared earlier 
as~\cite{NeeleVW2020}.}

\author[T.~Neele]{Thomas Neele\rsuper{a}}	
\address{\lsuper{a}Royal Holloway University of London, Egham, UK}	
\email{thomas.neele@rhul.ac.uk}  

\author[A.~Valmari]{Antti Valmari\rsuper{b}}	
\address{\lsuper{b}University of Jyv\"askyl\"a, Jyv\"askyl\"a, Finland}	
\email{antti.valmari@jyu.fi}  

\author[T.A.C.~Willemse]{Tim A.C. Willemse\rsuper{c}}	
\address{\lsuper{c}Eindhoven University of Technology, Eindhoven, The Netherlands}	
\email{t.a.c.willemse@tue.nl}  

\begin{abstract}
	One of the most popular state-space reduction techniques for model checking
	is partial-order reduction (POR).
	Of the many different POR implementations, stubborn sets are a very
	versatile variant and have thus seen many different applications over the
	past 32 years.
	One of the early stubborn sets works shows how the basic conditions for
	reduction can be augmented to preserve stutter-trace equivalence, making
	stubborn sets suitable for model checking of linear-time properties.
	In this paper, we identify a flaw in the reasoning and show with a
	counter-example that stutter-trace equivalence is not necessarily preserved.
	We propose a stronger reduction condition and provide extensive new
	correctness proofs to ensure the issue is resolved.
	Furthermore, we analyse in which formalisms the problem may occur.
	The impact on practical implementations is limited, since they all compute
	a correct approximation of the theory.
\end{abstract}

\maketitle

\section{Introduction}
In the field of formal methods, model checking is a push-button technique for
establishing the correctness of systems according to certain criteria.
A fundamental issue in model checking is the \emph{state-space explosion}
problem: the size of the state space can grow exponentially with the number of
concurrent components, due to all their possible interleavings.
One of the prime methods of reducing the number of states is
\emph{partial-order reduction} (POR).
The literature contains many different implementations of POR, but they are all
centred around the idea that some interleavings may be considered similar and
thus only one interleaving from each equivalence class needs to be explored.
The main variants of POR are \emph{ample sets}~\cite{Peled1993},
\emph{persistent sets}~\cite{Godefroid1996} and \emph{stubborn
sets}~\cite{Valmari1991a,Valmari2017a}.
The basic conditions set out by each of these variants can be strengthened,
such that the resulting conditions are sufficient for the preservation of
stutter-trace equivalence.
The extra conditions resolve the so-called \emph{action-ignoring
problem}~\cite{Valmari1991a}.
Since LTL without the next operator (LTL$_{-X}$) is invariant under finite
stuttering, this allows one to check most LTL properties under POR.

However, the correctness proofs for these methods are intricate and not
reproduced often.
For stubborn sets, LTL$_{-X}$-preserving conditions and an accompanying
correctness result were first presented in~\cite{Valmari1991b}; the
corresponding proofs appeared in~\cite{Valmari1992}.
When attempting to reproduce the proof of~\cite[Theorem 2]{Valmari1992} (see
also Theorem~\ref{thm:d1_preserve_stutter_trace_equivalence} in the current
work), we were unable to show that the two alternative paths considered
by~\cite[Construction 1]{Valmari1992}, a core component of the proof,
are stutter equivalent.
The consequence is that stutter-trace equivalence is not necessarily preserved,
contrary to what the theorem states!
We call this the \emph{inconsistent labelling problem}.

The essence of the problem is that POR in general, and the proofs
in~\cite{Valmari1992} in particular, reason mostly about actions, which label
the transitions.
In POR theory, the only relevance of the state labelling is that it determines
which actions must be considered \emph{visible}.
On the other hand, stutter-trace equivalence and the LTL semantics are purely
based on state labels.
The correctness proof in~\cite{Valmari1992} does not deal properly with this
disparity.
Consequently, any application of stubborn sets in LTL$_{-X}$ model checking is
possibly unsound, both for safety and liveness properties.
In literature, the correctness of several
theories~\cite{Laarman2016,Liebke2019,Valmari1996} relies on the incorrect
theorem.

In earlier work~\cite{NeeleVW2020}, we identified the inconsistent labelling
problem and investigated the theoretical and practical consequences.
As detailed in \ibid, the problem is witnessed by a counter-example, which is
valid for weak stubborn sets and, with a small modification, in a 
non-deterministic setting for strong stubborn sets.
A slight strengthening of one of the stubborn set conditions is sufficient to
repair the issue (Theorems~\ref{thm:dl_nostut} and~\ref{thm:inf_nostut} in the 
current work).
The fix is local, in the sense that it reduces the reduction potential in those 
places where the inconsistent labelling problem might otherwise occur.
Petri nets can be susceptible to the issue, depending on what notion of 
invisibility and what types of atomic propositions are used.
We used this knowledge about formalisms in which the inconsistent labelling
problem may manifest itself to determine its impact on related work.
The investigation in~\cite{NeeleVW2020} shows that probably all practical
implementations of stubborn sets compute an approximation which resolves the
inconsistent labelling problem.
Furthermore, POR methods based on the standard independence relation, such as
ample sets and persistent sets, are not affected.
The current paper improves on~\cite{NeeleVW2020} with extended explanation and
full proofs.
In particular, we introduce each of the existing stubborn set conditions with
reworked proofs to aid the reader's intuition.

The rest of the paper is structured as follows.
In Section~\ref{sec:preliminaries}, we introduce the basic concepts of
transition systems and stutter-trace equivalence.
Section~\ref{sec:stubborn_sets} introduces the stubborn set conditions one by
one and shows what they preserve through several lemmata.
Our counter-example to the preservation of stutter-trace equivalence is
presented in Section~\ref{sec:counter_example}.
We propose a solution to the inconsistent labelling problem in
Section~\ref{sec:strengthen_d1}, together with an updated correctness proof.
Sections~\ref{sec:safe_formalisms} and~\ref{sec:petri_nets} discuss several
settings in which correctness is not affected.
Finally, Section~\ref{sec:related_work} discusses related work and
Section~\ref{sec:conclusion} presents a conclusion.

\section{Preliminaries}
\label{sec:preliminaries}

\subsection{Labelled State Transition Systems and Paths}\label{sec:LSTS}

Since LTL relies on state labels and POR relies on edge labels, we assume the 
existence of some fixed set of atomic propositions \AP to label the states and 
a fixed set of edge labels \Act, which we will call \emph{actions}.
Actions are typically denoted with the letter $\act$.

\begin{defi}
	\label{def:labelled_LSTS}
	A \emph{labelled state transition system}, short \emph{LSTS}, is a directed 
	graph $\TS = (S,\edgerel,\init{s},L)$, where:
	\begin{itemize}
		\item $S$ is the state space;
		\item ${\edgerel} \subseteq S \times \Act \times S$ is the transition 
		relation;
		\item $\init{s} \in S$ is the initial state; and
		\item $L: S \to 2^{\AP}$ is a function that labels states with atomic 
		propositions.
	\end{itemize}
\end{defi}

We write $s \transition{\act} s'$ whenever $(s,\act,s') \in {\edgerel}$.
An action $\act$ is \emph{enabled} in a state $s$, notation $s
\transition{\act}$, if and only if there is a transition $s \transition{\act}
s'$ for some $s'$.
In a given LSTS \TS, $\enabled_{\TS}(s)$ is the set of all enabled actions in
a state $s$.
We may drop the subscript $\TS$ if it is clear from the context.
A state $s$ is a \emph{deadlock} in $\TS$ if and only if $\enabled_{\TS}(s) =
\emptyset$.

A \emph{path} is a (finite or infinite) alternating sequence of states and 
actions that respects the transition relation: $s_0 \transition{\act_1} s_1 
\transition{\act_2} s_2 \dots$.
We sometimes omit the intermediate and/or final states if they are clear from 
the context or not relevant, and write $s \transition{\act_1\dots \act_n} s'$
or $s \transition{\act_1\dots \act_n}$ for finite paths and $s
\transition{\act_1\act_2\dots}$ for infinite paths.
The empty sequence is denoted with $\varepsilon$.
Thus, for all states $s$ and $s'$, $s \transition{\varepsilon} s'$ holds if
and only if $s = s'$.
A path is \emph{deadlocking} if and only if it ends in a deadlock.
A path is \emph{complete} if and only if it is infinite or deadlocking.
Paths that start in the initial state $\init{s}$ are called \emph{initial 
paths}.

Given a path $\pi = s_0 \transition{\act_1} s_1 \transition{\act_2} s_2 \dots$, 
the \emph{trace} of $\pi$ is the sequence of state labels observed along $\pi$, 
\viz $L(s_0) L(s_1) L(s_2) \dots$.
The \emph{no-stutter trace} of $\pi$, notation $\nostut(\pi)$, is the 
sequence of those $L(s_i)$ such that $i = 0$ or $L(s_i) \neq L(s_{i-1})$.

A set $\Inv$ of \emph{invisible} actions is chosen such that if (but not 
necessarily only if) $a \in \Inv$, then for all states $s$ and $s'$, $s 
\transition{\act} s'$ implies $L(s) = L(s')$.
Note that this definition allows the set $\Inv$ to be under-approximated.
An action that is not invisible is called \emph{visible}.
The \emph{projection} of $a_1 \ldots a_n$ on the visible actions is the result
of the removal of all elements of $\Inv$ from $a_1 \ldots a_n$.
We denote it with $\vis_\Inv(a_1 \ldots a_n)$.
The notion extends naturally to infinite sequences $a_1 a_2 \ldots$.
We furthermore lift the function $\vis$ to paths, such that $\vis_\Inv(s_0 
\transition{a_1} s_1 \transition{a_2} \dots) = \vis_\Inv(a_1 a_2 \dots)$.
The subscript $\Inv$ is omitted when it is clear from the context.

We say \TS is \emph{deterministic} if and only if $s \transition{a} s_1$ and
$s \transition{a} s_2$ imply $s_1 = s_2$, for all states $s$, $s_1$ and $s_2$
and actions $\act$.
To indicate that \TS is not necessarily deterministic, we say \TS is 
\emph{non-deterministic}.

\subsection{Petri Nets}\label{section:PN}

\emph{Petri nets} are a widely-known formalism for modelling concurrent
processes and have seen frequent use in the application of stubborn set 
theory~\cite{Bonneland2019,Liebke2019,Valmari2017a,Varpaaniemi2005}.
We will use Petri nets for presenting examples.
In Section~\ref{sec:petri_nets}, we will also reassess the correctness of some 
published POR theories that use Petri nets.
Other than that, the theory in the present paper is fairly general, that is,
it does not depend on Petri Nets.

A Petri net $(P, T, W, \hat m)$ contains a set of \emph{places} $P$ and a set
of \emph{structural transitions} $T$.
These sets are disjoint.
In this paper they are finite.
Figure~\ref{fig:pn_example} shows an example of a Petri net.
Places are drawn as circles and structural transitions as rectangles.

\emph{Arcs} between places and structural transitions and their \emph{weights}
are specified via a total function $W: (P \times T) \cup (T \times P) \to
\mathbb{N}$.
The values $W(p,t)$ and $W(t,p)$ are called weights.
There is an arc from place $p$ to structural transition $t$, drawn as an
arrow, if and only if $W(p,t) > 0$; and similarly in the opposite direction if
and only if $W(t,p) > 0$.
If $W(p,t) > 1$ or $W(t,p) > 1$, then it is written as a number next to the
arc.
Figure~\ref{fig:pn_example} contains 11 arcs of weight 1, three arcs of
weight 2, and one arc of weight 3.

A \emph{marking} $m: P \to \mathbb{N}$ is a function that assigns a number of 
\emph{tokens} to each place.
Let $\Markings$ denote the set of all markings.
A Petri net has an \emph{initial marking} $\hat m$.
The initial marking of the example satisfies $\hat m(p_3) = 2$, $\hat m(p_1) =
\hat m(p_4) = \hat m(p_6) = 1$ and $\hat m(p_2) = \hat m(p_5) = 0$.

Structural transition $t$ is \emph{enabled} in marking $m$ if and only if
$m(p) \geq W(p,\tr)$ for every $p \in P$, and \emph{disabled} otherwise.
In our example, $t_1$, $t_3$ and $t_6$ are enabled.
Because $\hat m(p_3) = 2$ but $W(p_3, t_4) = 3$, $t_4$ is disabled.
An enabled transition may \emph{occur} resulting in the marking $m'$ such that
$m'(p) = m(p) - W(p,\tr) + W(\tr,p)$ for every $p \in P$.
We denote this with $m \transition{\tr} m'$, and extend the notation to paths
similarly to Section~\ref{sec:LSTS}.
If $m$ is the marking such that $\hat m \transition{t_1} m$ in our example,
then $m(p_1) = 0$, $m(p_2) = 1$, and $m(p) = \hat m(p)$ for the remaining
places.
If $\hat m \transition{t_3} m'$, then $m'(p_4) = 0$ and $m'(p) = \hat m(p)$
for the remaining places.

A marking $m$ is \emph{reachable} if and only if there are $t_1$, \ldots,
$t_n$ such that $\hat{m} \transition{t_1 \ldots t_n} m$.
Let $\Mreach$ denote the set of reachable markings, and $\edgerel'$ the
restriction of $\edgerel$ on $\Mreach \times T \times \Mreach$.
Assume that a set of atomic propositions $\AP$ and a function $L': \Mreach
\to 2^\AP$ are given.
A Petri net together with these induces the LSTS $(\Mreach, \edgerel', \hat
m, L')$.
In this context $\Act = T$.

It is customary to abuse notation by forgetting about the distinction between
$\edgerel$ and $\edgerel'$, and using the same symbol for both.
This is done because it is often not known in advance whether a marking is
reachable, making it impractical to define $\edgerel'$ instead of $\edgerel$.
Similarly instead of $L'$, it is customary to define a function $L$ from all
markings $\Markings$ to $2^\AP$, let $L'$ be its restriction on $\Mreach$, and 
abuse notation by using the same symbol for both.
These are general practice instead of being restricted to Petri nets.

\begin{figure}
	\centering
	\begin{tikzpicture}[->,>=stealth',shorten >=0pt,auto,node 
	distance=2.0cm,semithick,
	every place/.style={draw,minimum size=5mm}]
	
	\begin{scope}
	\tikzstyle{vtransition} = [fill,inner sep=0pt,minimum 
	width=1.4mm,minimum height=5mm]
	\tikzstyle{htransition} = [fill,inner sep=0pt,minimum 
	width=5mm,minimum height=1.4mm]
	
	\def\x{1}
	\def\y{1.2}
	\node[place,label={above:$p_1$},tokens=1]       (p1) at (0,\y) {};
	\node[place,label={above:$p_2$}]                (p2) at (2*\x,\y) {};
	\node[place,label={above:$p_3$},tokens=2]       (p3) at (4*\x,\y) {};
	\node[place,label={above:$p_4$},tokens=1]       (p4) at (6*\x,\y) {};
	\node[place,label={above:$p_5$}]                (p5) at (8*\x,\y) {};
	\node[place,label={above:$p_6$},tokens=1]       (p6) at (10*\x,\y) {};
	
	\node[vtransition,label={above:$\tr_1$}]  (t1) at (1*\x,\y) {};
	\node[vtransition,label={above:$\tr_2$}]  (t2) at (3*\x,\y) {};
	\node[vtransition,label={above:$\tr_3$}]  (t3) at (5*\x,\y) {};
	\node[vtransition,label={above:$\tr_5$}]  (t5) at (7*\x,\y) {};
	\node[vtransition,label={above:$\tr_6$}]  (t6) at (9*\x,\y) {};
	\node[htransition,label={left:$\tr_4$}]   (t4) at (4*\x,0) {};
	
	\path
	(p1) edge (t1) (t1) edge (p2)
	(p2) edge (t2) (t2) edge (p3)
	(p3) edge[bend left=10] node {2} (t3.150)
	(t3.210) edge[bend left=10] node {2} (p3)
	(p3) edge[bend left=10] node {3} (t4.60)
	(t4.120) edge[bend left=10] node {2} (p3)
	(p4) edge (t3) (p4) edge (t5)
	(p5) edge (t5) (t6) edge (p5)
	(p6) edge[bend left=10] (t6.330) (t6.30) edge[bend left=10] (p6);
	\draw[rounded corners=6pt]
	(p6) |- (t4);
	
	\end{scope}
	\end{tikzpicture}
	\caption{An example Petri net.}
	\label{fig:pn_example}
\end{figure}
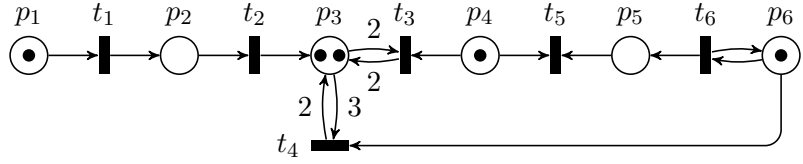

\subsection{Weak and Stutter Equivalence}\label{sec:weak_stutter}

Stubborn sets save effort by constructing, instead of the full LSTS $\TS =
(S,\edgerel,\init{s},L)$, a reduced LSTS $\TS_r = (S_r, \edgerel_r, \init{s},
L_r)$ such that $S_r \subseteq S$, ${\edgerel_r} \subseteq {\edgerel}$ and
$L_r$ is the restriction of $L$ on $S_r$ (more details will be given in
Section~\ref{sec:stubborn_sets}).
To reason about the similarity of an LSTS $\TS$ and its reduced LSTS $\TS_r$, 
we introduce the notions \emph{weak equivalence}, which operates on actions, 
and \emph{stutter equivalence}, which operates on states.
For the purpose of the discussion in Section~\ref{sec:petri_nets}, these 
concepts respectively depend on a set of actions and a labelling function.

\begin{defi}
	Two paths $\pi$ and $\pi'$ are weakly equivalent with respect to a set of 
	actions $A$, notation $\pi \weakeq_A \pi'$, if and only if they are both 
	finite or both infinite, and their respective projections on $\Act 
	\setminus A$ are equal, \ie, $\vis_A(\pi) = \vis_A(\pi')$.
\end{defi}
\begin{defi}
	\label{def:stutter_equivalence}
	Paths $\pi$ and $\pi'$ are stutter equivalent under $L$, notation $\pi 
	\stuteq_L \pi'$, if and only if they are both finite or both infinite, and 
	they yield the same no-stutter trace under~$L$.
\end{defi}

We typically consider weak equivalence with respect to the set of invisible 
actions $\Inv$.
In that case, we simply refer to the equivalence as weak equivalence and
we write $\pi \weakeq \pi'$, which intuitively means that $\pi$ and $\pi'$ 
contain the same visible actions.
We also omit the subscript for stutter equivalence when reasoning about the 
labelling function of the LSTS under consideration and write $\pi \stuteq \pi'$.
Note that stutter equivalence is invariant under finite repetitions of state 
labels, hence its name.
We lift both equivalences to LSTSs, and say that $\TS$ and $\TS'$ are 
\emph{weak-trace equivalent}  iff for every complete initial path $\pi$ in
$\TS$, there is a weakly equivalent complete initial path $\pi'$ in $\TS'$ and
vice versa.
Likewise, $\TS$ and $\TS'$ are \emph{stutter-trace equivalent} iff for every 
complete initial path $\pi$ in $\TS$, there is a stutter equivalent complete
initial path $\pi'$ in $\TS'$ and vice versa.

In general, weak equivalence and stutter equivalence are incomparable, even for 
complete initial paths.
However, for some LSTSs, these notions \emph{are} related in a certain way.
We formalise this in the following definition.
\begin{defi}
	\label{def:consistent_labelling}
	An LSTS is \emph{labelled consistently} iff for all complete initial
paths $\pi$ and 
	$\pi'$, $\pi \weakeq \pi'$ implies $\pi \stuteq \pi'$.
\end{defi}
It follows from the definition that, if an LSTS $\TS$ is labelled consistently 
and weak-trace equivalent to a subgraph $\TS'$, then $\TS$ and $\TS'$ are also 
stutter-trace equivalent.

Stubborn sets as defined in the next section aim to preserve stutter-trace 
equivalence between the original and the reduced LSTS.
The motivation behind this is that two stutter-trace equivalent LSTSs satisfy 
exactly the same formulae~\cite{BaierKatoen-PMC} in LTL$_{-X}$.
The following theorem, which is frequently cited in the 
literature~\cite{Laarman2016,Liebke2019,Valmari1996}, aims to show that 
stubborn sets indeed preserve stutter-trace equivalence.
Its original formulation reasons about the validity of an arbitrary LTL$_{-X}$ 
formula.
Here, we give the alternative formulation based on stutter-trace equivalence.

\begin{thm}{\cite[Theorem 2]{Valmari1992}}
	\label{thm:d1_preserve_stutter_trace_equivalence}
	For every LSTS $\TS$, the reduced LSTS $\TS_\redf$ (defined in 
	Section~\ref{sec:stubborn_sets}) is stutter-trace equivalent to $\TS$.
\end{thm}


The original proof correctly establishes the four items listed below.
For a long time it was believed that they suffice to ensure that $\TS_r$ gives 
the same truth values to LTL$_{-X}$ formulas as $\TS$ gives.
While investigating the application of stubborn sets to parity 
games~\cite{NeeleWW2020}, Thomas Neele (the main author of the current paper, 
but not the author of this sentence) took the effort of checking this 
self-evident ``fact'', and found out that it does not hold.
We call this the \emph{inconsistent labelling problem}.
A counter-example is in Section~\ref{sec:counter_example}.
\begin{enumerate}
\item Every initial deadlocking path of $\TS$ has a weakly equivalent initial
deadlocking path in $\TS_r$.
\item Every initial deadlocking path of $\TS_r$ has a weakly equivalent
initial deadlocking path in $\TS$.
\item Every initial infinite path of $\TS$ has a weakly equivalent initial
infinite path in $\TS_r$.
\item Every initial infinite path of $\TS_r$ has a weakly equivalent initial
infinite path in $\TS$.
\end{enumerate}
Because the four items in this list are sufficient for $\TS \weakeq \TS_r$, the 
issue could be resolved with the additional requirement that $\TS$ is 
consistently labelled, which would yield $\TS \stuteq \TS_r$ (since $\TS_r$ is 
a subgraph of $\TS$, see Definition~\ref{def:red_lsts}).
However, this requirement is rather strong; we propose a more local solution in 
Section~\ref{sec:strengthen_d1}.

\section{Stubborn Sets}\label{sec:stubborn_sets}

\subsection{Basic Ideas}

In POR, \emph{reduction functions} play a central role.
A reduction function $\redf: S \to 2^\Act$ indicates which actions to explore
in each state.
When starting at the initial state $\init{s}$, a reduction function induces a 
\emph{reduced LSTS} as follows.

\begin{defi}
	\label{def:red_lsts}
	Let $\TS = (S,\edgerel,\init{s},L)$ be an LSTS and $\redf: S \to 2^\Act$ a 
	reduction function.
	Then the \emph{reduced LSTS} induced by $\redf$ is defined as $\TS_r = 
	(S_r,\rededgerel,\init{s},L_r)$, where $L_r$ is the restriction of $L$ on 
	$S_r$, and $S_r$ and $\rededgerel$ are the smallest sets such that the 
	following holds:
\begin{itemize}
\item $\init{s} \in S_r$; and
\item If $s \in S_r$, $s \transition{\act} s'$ and $\act \in \redf(s)$, then
$s' \in S_r$ and $s \redtransition{\act} s'$.
\end{itemize}
\end{defi}

Note that we have ${\rededgerel} \subseteq {\edgerel}$.\smallskip

In the first paper on stubborn sets~\cite{Valmari1988}, the set $r(s)$ was
constructed so that if enabled actions exist, then it contains an enabled
action that the outside world cannot disable.
This inspired the thought that the set is ``stubborn'', that is, determined to
do something and not letting the outside world prevent it.
Much more than this is needed to make $\TS_r$ yield correct answers to
verification questions concerning $\TS$.
Furthermore, some more recent methods do not necessarily put an enabled action
in $r(s)$ even if enabled actions do exist.
So the name is imprecise, but has remained in use.

The main question now is how to implement a practical reduction function so
that answers to interesting verification questions can be obtained from the
reduced LSTSs.
Because this publication is about fixing an error that had been lurking for
decades, we feel appropriate to present the full proof of the affected theorem
anew as clearly as possible, in more detail than originally, to minimise the
possibility that other errors remain.
To this end, we proceed in small steps.

We first discuss the motivating example from Figure~\ref{fig:pn_example}, 
reproduced here in Figure~\ref{fig:D1D2wExample}.
%
Assume that we know that the places adjacent to $t_3$ are $p_3$ and $p_4$;
they contain 2 and 1 tokens, respectively; the transitions adjacent to $p_3$
and $p_4$ are $t_2$ to $t_5$; and the arcs between them and their weights are
as is shown in Figure~\ref{fig:D1D2wExample}.
That is, we know the black part but not the grey part in the figure.
Although our knowledge is incomplete, we can reason as follows that $t_3$ is
enabled and remains enabled until $t_3$ or $t_5$ occurs.
It is enabled by the numbers of tokens in $p_3$ and $p_4$, and by the weights
of the arcs from them to $t_3$.
An occurrence of $t_2$ does not decrement the numbers of tokens in $p_3$ and
$p_4$, so it cannot disable $t_3$.
The same applies to $t_1$ and $t_6$.
An occurrence of $t_4$ decrements the number of tokens in $p_3$ (but not in
$p_4$).
However, thanks to the arc weight 2, it is guaranteed to leave at least 2
tokens in $p_3$.
So it cannot disable $t_3$ either.

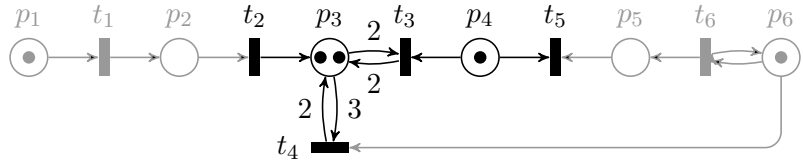
\begin{figure}
	\centering
	\begin{tikzpicture}[->,>=stealth',shorten >=0pt,auto,node 
	distance=2.0cm,semithick,
	every place/.style={draw,minimum size=5mm}]
	
	\begin{scope}
	\tikzstyle{vtransition} = [fill,inner sep=0pt,minimum 
	width=1.4mm,minimum height=5mm]
	\tikzstyle{htransition} = [fill,inner sep=0pt,minimum 
	width=5mm,minimum height=1.4mm]
	
	\def\x{1}
	\def\y{1.2}
	\colorlet{hiddencolor}{black!40}
	\node[place,label={[color=hiddencolor]above:$p_1$},hiddencolor]      (p1) 
	at (0,\y) {}
	[children are tokens] child {node[token,fill=hiddencolor] {}};
	\node[place,label={[color=hiddencolor]above:$p_2$},hiddencolor]      (p2) 
	at (2*\x,\y) {};
	\node[place,label={above:$p_3$},tokens=2]       (p3) at (4*\x,\y) {};
	\node[place,label={above:$p_4$},tokens=1]       (p4) at (6*\x,\y) {};
	\node[place,label={[color=hiddencolor]above:$p_5$},hiddencolor]      (p5) 
	at (8*\x,\y) {};
	\node[place,label={[color=hiddencolor]above:$p_6$},hiddencolor]      (p6) 
	at (10*\x,\y) {}
	[children are tokens] child {node[token,fill=hiddencolor] {}};
	
	\node[vtransition,label={[color=hiddencolor]above:$\tr_1$},fill=hiddencolor]
	(t1) at (1*\x,\y) {};
	\node[vtransition,label={above:$\tr_2$}]  (t2) at (3*\x,\y) {};
	\node[vtransition,label={above:$\tr_3$}]  (t3) at (5*\x,\y) {};
	\node[vtransition,label={above:$\tr_5$}]  (t5) at (7*\x,\y) {};
	\node[vtransition,label={[color=hiddencolor]above:$\tr_6$},fill=hiddencolor]
	(t6) at (9*\x,\y){};
	\node[htransition,label={left:$\tr_4$}]   (t4) at (4*\x,0) {};
	
	\path
	(t2) edge (p3)
	(p3) edge[bend left=10] node {2} (t3.150)
	(t3.210) edge[bend left=10] node {2} (p3)
	(p3) edge[bend left=10] node {3} (t4.60)
	(t4.120) edge[bend left=10] node {2} (p3)
	(p4) edge (t3) (p4) edge (t5)
	;
	\path[draw=hiddencolor]
	(p1) edge (t1) (t1) edge (p2)
	(p2) edge (t2)
	(p5) edge (t5) (t6) edge (p5)
	(p6) edge[bend left=10] (t6.330) (t6.30) edge[bend left=10] (p6)
	;
	\draw[rounded corners=6pt,hiddencolor]
	(p6) |- (t4);
	
	\end{scope}
	\end{tikzpicture}
	\caption{An example motivating \textbf{D1} and \textbf{D2w}.}
	\label{fig:D1D2wExample}
\end{figure}

This is an example of the kind of observations that stubborn set methods
exploit.
Together with some other observations that will be discussed soon, it will let
us choose $\redf(s) = \{t_3, t_5, t_6\}$, where $s$ denotes the marking shown
in Figure~\ref{fig:D1D2wExample}.
Unfortunately, the observation is Petri net-specific.
We now introduce a more abstract notion that captures the same idea: $t_3$ is
a \emph{key action} of $\redf(s) = \{t_3, t_5, t_6\}$ in the sense of the
following definition.
\begin{defi}
An action $\act$ is a \emph{key action} of $\redf(s)$ in $s$ if and only if
for all paths $s \transition{\act_1\dots\act_n} s'$ such that $\act_1 \notin
\redf(s), \dots, \act_n \notin \redf(s)$, it holds that $s'
\transition{\act}$.
\end{defi}
We typically denote key actions by $\keyact$.
Note that a key action must be enabled in $s$: by setting $n=0$, we have $s = 
s'$ and $s \transition{\act}$.

Many stubborn set methods assume that the sets $\redf(s)$ satisfy the
following condition.
\begin{description}[leftmargin=!,labelwidth=\widthof{\bfseries D2w:}]
	\item[\textbf{D2w}] If $\enabled(s) \neq \emptyset$, then $\redf(s)$ 
	contains a key action in $s$.
\end{description}

In Figure~\ref{fig:D1D2wExample}, $t_5$ is not a key action of $\{t_3, t_5,
t_6\}$, because it is disabled.
Also $t_6$ is not, because the sequence $t_1 t_2 t_4$ disables it.

\begin{figure}[t]
	\centering
	\resizebox{\textwidth}{!}{
		\begin{tikzpicture}[->,>=stealth',shorten >=0pt,auto,node 
		distance=2.0cm,semithick]
		\small
		\begin{scope}
			\def\d{1.4}
			\node (s)     at (0,\d)        {$s$};
			\node (s1)    at (\d,\d)       {$s_1$};
			\node (d)     at (1.65*\d,\d)  {$\dots$};
			\node (sn-1)  at (2.5*\d,\d)   {$s_{n-1}$};
			\node (sn)    at (3.5*\d,\d)   {$s_n$};
			\node (s'n)   at (3.5*\d,0)    {$s'_n$};
			\node         at (1.65*\d,\d+0.8) {$\notin \redf(s)$};
			\draw[-] decorate[decoration={name=brace,amplitude=5pt}] {
				($(s.north east) + (0,0.2)$) -- ($(sn.north west) + (0,0.2)$)
			};
			\path
			(s)     edge node {$\act_1$} (s1)
			(d)     edge              (sn-1)
			(sn-1)  edge node {$\act_n$} (sn)
			(sn)    edge node {$\act \in \redf(s)$}   (s'n);
			\path[-]
			(s1)    edge              (d);
		\end{scope}
		\node at (6.9,0.7) {$\Rightarrow$};
		\begin{scope}[xshift=7.7cm]
			\def\d{1.4}
			\node (s)     at (0,\d)        {$s$};
			\node (s1)    at (\d,\d)       {$s_1$};
			\node (d)     at (1.65*\d,\d)  {$\dots$};
			\node (sn-1)  at (2.5*\d,\d)   {$s_{n-1}$};
			\node (sn)    at (3.5*\d,\d)   {$s_n$};
			\node (s')    at (0,0)         {$s'$};
			\node (s'1)   at (\d,0)        {$s'_1$};
			\node (s'n-1) at (2.5*\d,0)    {$s'_{n-1}$};
			\node (s'n)   at (3.5*\d,0)    {$s'_n$};
			\path
			(s1)    edge node {$\act$}   (s'1)
			(sn-1)    edge node {$\act$}   (s'n-1)
			(s)     edge node {$\act_1$} (s1)
			(d)     edge              (sn-1)
			(sn-1)  edge node {$\act_n$} (sn)
			(sn)    edge node {$\act$}   (s'n);
			\path[-]
			(s1)    edge              (d);
			\path
			(s)     edge node {$\act$}   (s');
		\end{scope}
		\node at (13.6,0.7) {$\Rightarrow$};
		\end{tikzpicture}
	}
	\resizebox{\textwidth}{!}{
		\begin{tikzpicture}[->,>=stealth',shorten >=0pt,auto,node 
		distance=2.0cm,semithick]
		\small
		\begin{scope}
			\def\d{1.4}
			\node (s)     at (0,\d)        {$s$};
			\node (s1)    at (\d,\d)       {$s_1$};
			\node (d)     at (1.65*\d,\d)  {$\dots$};
			\node (sn-1)  at (2.5*\d,\d)   {$s_{n-1}$};
			\node (sn)    at (3.5*\d,\d)   {$s_n$};
			\node (s')    at (0,0)         {$s'$};
			\node (s'1)   at (\d,0)        {$s'_1$};
			\node (s''1)  at (\d-.2,-.5)   {$s''_1$};
			\node (s''2)  at (2*\d-.2,-.5) {$s''_2$};
			\node (d'')    at (2.5*\d-.1,-.5)  {$\dots$};
			\node (s'n-1) at (2.5*\d,0)    {$s'_{n-1}$};
			\node (s'n)   at (3.5*\d,0)    {$s'_n$};
			\node (s''n)   at (3.5*\d-.2,-.5) {$s''_n$};
			\path
			(s)     edge node {$\act$}   (s')
			(s1)    edge node {$\act$}   (s'1)
			(sn-1)  edge node {$\act$}   (s'n-1)
			(sn)    edge node {$\act$}   (s'n)
			(s)     edge node {$\act_1$} (s1)
			(d)     edge              (sn-1)
			(sn-1)  edge node {$\act_n$} (sn)
			(s')    edge node {$\!\!\!\act_1$} (s''1)
			(s'1)   edge node {$\!\!\!\act_2$} (s''2)
			(s'n-1) edge node {$\!\!\!\act_n$} (s''n);
			\path[-]
			(s1)    edge              (d);
		\end{scope}
		\node at (6.9,0.7) {$\Rightarrow$};
		\begin{scope}[xshift=7.7cm]
			\def\d{1.4}
			\node (s)     at (0,\d)        {$s$};
			\node (s1)    at (\d,\d)       {$s_1$};
			\node (d)     at (1.65*\d,\d)  {$\dots$};
			\node (sn-1)  at (2.5*\d,\d)   {$s_{n-1}$};
			\node (sn)    at (3.5*\d,\d)   {$s_n$};
			\node (s')    at (0,0)         {$s'$};
			\node (s'1)   at (\d,0)        {$s'_1$};
			\node (d')    at (1.65*\d,0)   {$\dots$};
			\node (s'n-1) at (2.5*\d,0)    {$s'_{n-1}$};
			\node (s'n)   at (3.5*\d,0)    {$s'_n$};
			\path
			(s1)    edge node {$\act$}   (s'1)
			(sn-1)    edge node {$\act$}   (s'n-1)
			(s)     edge node {$\act_1$} (s1)
			(d)     edge              (sn-1)
			(sn-1)  edge node {$\act_n$} (sn)
			(sn)    edge node {$\act$}   (s'n);
			\path[-]
			(s1)    edge              (d);
			\path
			(s')    edge node {$\act_1$} (s'1)
			(d')    edge              (s'n-1)
			(s'n-1) edge node {$\act_n$} (s'n)
			(s)     edge node {$\act$}   (s');
			\path[-]
			(s'1)   edge              (d');
		\end{scope}
		\node at (13.6,0.7) {\phantom{$\Rightarrow$}};
		\end{tikzpicture}
	}
	\caption{Visual representation of why \textbf{D1} holds on the example.}
	\label{fig:PN_D1}
\end{figure}
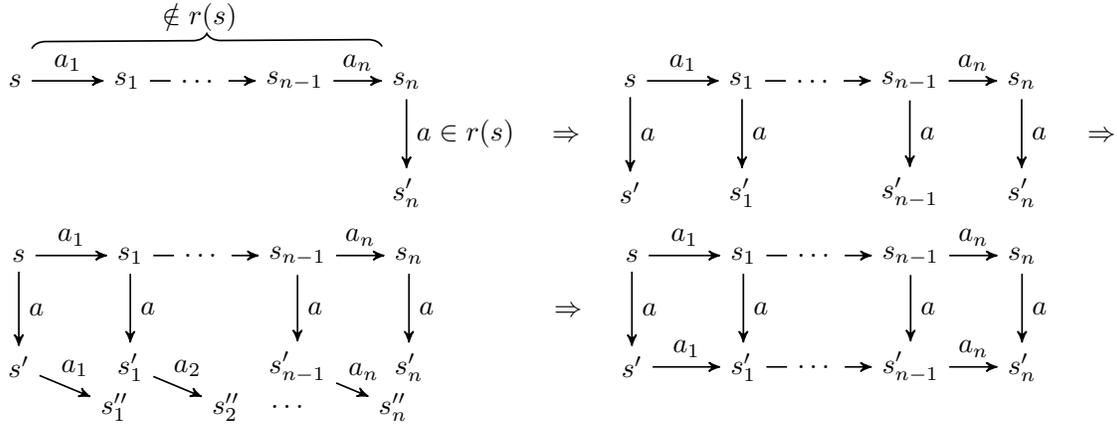

On the other hand, we now show that $t_3$, $t_5$ and $t_6$ have another
property that stubborn set methods exploit: Figure~\ref{fig:PN_D1} holds for
each of them in the role of $a$ and any finite sequence of elements of $\{t_1,
t_2, t_4\}$ in the role of $a_1 \ldots a_n$.
We call $t_1$, $t_2$ and $t_4$ the outside transitions.

Although the outside transitions can disable $t_6$, they cannot enable it
again, because none of them can add tokens to $p_6$.
Therefore, if $t_6$ is enabled after the occurrence of some sequence $a_1
\ldots a_n$ of outside transitions, then it was enabled in the original
marking $s$ and in every marking between $s$ and $s_n$.
This is illustrated by the first implication in Figure~\ref{fig:PN_D1}, with
$t_6$ in the role of $a$.
The first implication applies to $t_3$ as well, because its right-hand side
applies, because $t_3$ is a key action of $\{t_3, t_5, t_6\}$ in $s$.

Neither $t_3$ nor $t_6$ can disable outside transitions, because although they
temporarily consume tokens from $p_3$ or $p_6$, they put the same number of
tokens back to them; and the outside transitions do not need tokens from
$p_4$.
This yields the second implication in the figure.
Furthermore, Petri nets are commutative in the sense that if $m \transition{t'
t} m'$ and $m \transition{t t'} m''$, then $m' = m''$.
The last implication in the figure holds because of this.

The implication chain also applies to $t_5$ as $a$, but for a different
reason: $t_5$ is disabled, and no sequence of outside transitions can enable
it, because only $t_6$ can enable it.
Therefore, no member of the chain holds for $t_5$, so the chain holds
vacuously.

Again, we appealed to particular properties of Petri nets.
To make the ideas applicable to a wide variety of formalisms for representing
systems, we introduce the following condition, which is required to hold for 
all $\redf(s)$.
It is illustrated in Figure~\ref{fig:condition_D1}.
We showed above that it holds for $\redf(s) = \{t_3, t_5, t_6\}$ in
Figure~\ref{fig:D1D2wExample}.

\begin{description}[leftmargin=!,labelwidth=\widthof{\bfseries D1:}]
	\item[\textbf{D1}] For all states $s_1,\dots,s_n,s'_n$ and all $\act 
	\in \redf(s)$ and $\act_1 \notin \redf(s), \dots, \act_n \notin \redf(s)$, 
	if $s \transition{\act_1} \dots \transition{\act_n} s_n \transition{\act} 
	s'_n$, then there are states $s',s'_1,\dots,s'_{n-1}$ such that $s 
	\transition{\act} s' \transition{\act_1} s'_1 \transition{\act_2} \dots 
	\transition{\act_n} s'_n$.
\end{description}

\begin{figure}[t]
	\centering
	\resizebox{\textwidth}{!}{
		\begin{tikzpicture}[->,>=stealth',shorten >=0pt,auto,node 
		distance=2.0cm,semithick]
		\small
		\begin{scope}
			\def\d{1.4}
			\node (s)     at (0,\d)        {$s$};
			\node (s1)    at (\d,\d)       {$s_1$};
			\node (d)     at (1.65*\d,\d)  {$\dots$};
			\node (sn-1)  at (2.5*\d,\d)   {$s_{n-1}$};
			\node (sn)    at (3.5*\d,\d)   {$s_n$};
			\node (s'n)   at (3.5*\d,0)    {$s'_n$};
			\node         at (1.65*\d,\d+0.8) {$\notin \redf(s)$};
			\draw[-] decorate[decoration={name=brace,amplitude=5pt}] {
				($(s.north east) + (0,0.2)$) -- ($(sn.north west) + (0,0.2)$)
			};
			\path
			(s)     edge node {$\act_1$} (s1)
			(d)     edge              (sn-1)
			(sn-1)  edge node {$\act_n$} (sn)
			(sn)    edge node {$\act \in \redf(s)$}   (s'n);
			\path[-]
			(s1)    edge              (d);
		\end{scope}
		\node at (6.9,0.7) {$\Rightarrow$};
		\begin{scope}[xshift=7.7cm]
			\def\d{1.4}
			\node (s)     at (0,\d)        {$s$};
			\node (s1)    at (\d,\d)       {$s_1$};
			\node (d)     at (1.65*\d,\d)  {$\dots$};
			\node (sn-1)  at (2.5*\d,\d)   {$s_{n-1}$};
			\node (sn)    at (3.5*\d,\d)   {$s_n$};
			\node (s')    at (0,0)         {$s'$};
			\node (s'1)   at (\d,0)        {$s'_1$};
			\node (d')    at (1.65*\d,0)   {$\dots$};
			\node (s'n-1) at (2.5*\d,0)    {$s'_{n-1}$};
			\node (s'n)   at (3.5*\d,0)    {$s'_n$};
			\path
			(s)     edge node {$\act_1$} (s1)
			(d)     edge              (sn-1)
			(sn-1)  edge node {$\act_n$} (sn)
			(sn)    edge node {$\act$}   (s'n);
			\path[-]
			(s1)    edge              (d);
			\path
			(s')    edge node {$\act_1$} (s'1)
			(d')    edge              (s'n-1)
			(s'n-1) edge node {$\act_n$} (s'n)
			(s)     edge node {$\act$}   (s');
			\path[-]
			(s'1)   edge              (d');
		\end{scope}
		\end{tikzpicture}
	}
	\caption{Visual representation of condition \textbf{D1}.}
	\label{fig:condition_D1}
\end{figure}
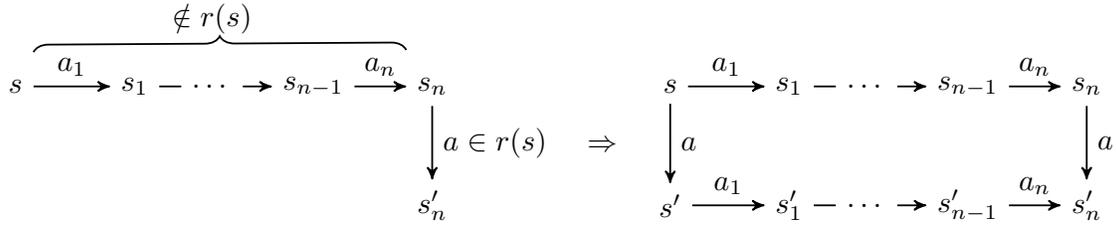

\subsection{Deadlock Detection and Its Implementation}
\label{sec:stub_deadlock}

The conditions \textbf{D2w} and \textbf{D1} are important, because they
suffice for proving that all reachable deadlocks of the full LSTS are present 
also in the reduced LSTS.
Furthermore, the deadlocks can be reached in the reduced LSTS by re-ordering
the actions in the paths in the full LSTS that lead to them.
In the theory below, recall that $\rededgerel$ indicates which transitions 
occur in the reduced LSTS.

\begin{thm}\label{thm:deadlock}
Assume that each $\redf(s)$ obeys \textbf{D1} and \textbf{D2w}.
If $s_0 \in S_r$, $s_n$ is a deadlock in $\TS$, and $s_0 \transition{a_1} s_1
\transition{a_2} \ldots \transition{a_n} s_n$, then there is a permutation
$b_1 \ldots b_n$ of $a_1 \ldots a_n$ such that $s_0 \redtransition{b_1 \ldots
b_n} s_n$.
\end{thm}
\begin{proof}
We prove the claim by induction on $n$.
If $n = 0$ then $s_n = s_0$ and $a_1 \ldots a_n = \varepsilon$, so the claim
holds trivially with $b_1 \ldots b_n = \varepsilon$.

From now on, let $n > 0$.
We have $s_0 \transition{a_1}$ and thus $\enabled(s_0) \neq \emptyset$.
By \textbf{D2w}, $\redf(s_0)$ contains a key action $\keyact$.
If none of $a_1$, \ldots, $a_n$ is in $\redf(s_0)$, then by definition $s_n
\transition{\keyact}$.
However, that cannot be the case, because we assumed that $s_n$ is a deadlock.
Therefore, there is $1 \leq i \leq n$ such that $a_i \in \redf(s_0)$.
We choose the smallest such $i$, yielding $a_j \notin \redf(s_0)$ for $1 \leq
j < i$.
By this choice, \textbf{D1} applies with $a_i$ in the role of $a$.
So there are states $s'_0$, $s'_1$, \dots, $s'_{i-1}$ such that $s'_{i-1} =
s_i$ and $s_0 \transition{\act_i} s'_0 \transition{\act_1} s'_1
\transition{\act_2} \dots \transition{\act_{i-1}} s'_{i-1}$.
Because $a_i \in \redf(s_0)$ we have $s'_0 \in S_r$ and $(s_0, a_i, s'_0) \in
{\rededgerel}$, that is, $s_0 \redtransition{\act_i} s'_0$.
We remember that $s_i \transition{a_{i+1} \ldots a_n} s_n$, so $s'_0
\transition{a_1 \ldots a_{i-1} a_{i+1} \ldots a_n} s_n$.
This path is one shorter than $s_0 \transition{a_1 \ldots a_n} s_n$.
Therefore, the induction assumption yields a permutation $b_2 \ldots b_n$ of
$a_1 \ldots a_{i-1} a_{i+1} \ldots a_n$ such that $s'_0 \redtransition{b_2 
\ldots b_n} s_n$.
As a consequence, $s_0 \redtransition{a_i b_2 \ldots b_n} s_n$.
\end{proof}

The preservation of deadlocks needs also the following facts, which are easy
to check from the definitions.
\begin{itemize}
\item If $s_0 \redtransition{a_1 \ldots a_n} s_n$, then $s_0 \transition{a_1
\ldots a_n} s_n$.
\item $s \in S_r$ is a deadlock in $\TS_r$ if and only if it is a deadlock in
$\TS$.
\end{itemize}

To implement this deadlock detection method, an algorithm is needed that,
given state~$s$, computes a set $\redf(s)$ that satisfies \textbf{D1} and
\textbf{D2w}.
We already illustrated with Figure~\ref{fig:D1D2wExample} that this may depend
on the details of the formalism used to represent the system under
verification.
Because it is sometimes very difficult to check whether \textbf{D1} and
\textbf{D2w} hold, the algorithms rely on formalism-specific heuristics that
may give a false negative but cannot give a false positive.
The set of all actions satisfies \textbf{D1} and \textbf{D2w}.
While it yields no reduction, it can be used as a fall-back when attempts to
find a better set fail.

The algorithm design problem also involves a trade-off between the time it
takes to compute the set, and the quality of the set: smaller sets tend to
result in smaller reduced LSTSs (although this issue is not
straightforward~\cite{Valmari2017a}).
In the case of 1-safe Petri nets, testing whether a singleton set $\{t\}$ is a
valid $\redf(s)$ for the purpose of preserving all deadlocks is
\textbf{PSPACE}-hard~\cite{ValmariH11}.
This means that there is not much hope of a fast algorithm that always yields
the best possible $r(s)$.

Instead, algorithms range from quick and simple that exploit only the most
obvious reduction possibilities, to very complicated that spend unreasonable
amounts of time and memory in trying to find a set with few enabled actions.
For instance, after finding out that $t_5$ may disable $t_3$ in
Figure~\ref{fig:D1D2wExample}, $\{t_3\}$ must be rejected as a candidate
$\redf(s)$.
A simple algorithm might revert to the set of all actions, while a more
complicated algorithm might try $\{t_3, t_5\}$, detect that $t_6$ might enable
$t_5$, try $\{t_3, t_5, t_6\}$, and find out that it works.

Fortunately, the kind of analysis that led us from $\{t_3\}$ to $\{t_3, t_5,
t_6\}$ is not at all too expensive, if we are okay with some imperfection.
It can be performed in linear time by formulating it as the problem of finding
certain kinds of maximal strongly connected components in a directed graph
whose edges $t \leadsto t'$ represent the notion ``if $t \in \redf(s)$, then
also $t' \in \redf(s)$'' (e.g.,~\cite{Valmari2017a}).
The result is optimal in a sense that is meaningful albeit not
perfect~\cite{ValmariH11}.
(In the light of \textbf{PSPACE}-hardness, we should not expect perfection.)

One of the things that it cannot optimise is which enabled action to choose as
a key action, if many are available.
In our example, it would have been possible to choose $t_1$ or $t_6$ instead
of $t_3$.
Because $t_6$ may be disabled by $t_4$, which is disabled until $t_2$ occurs,
which is disabled until $t_1$ occurs, the choice of $t_6$ introduces the edges
$t_6 \leadsto t_4 \leadsto t_2 \leadsto t_1$.
The resulting $\redf(s)$ would be $\{t_1\}$, because $t_1$ is enabled and does
not compete for tokens with any other transition.
That is, the algorithm is clever enough to drop $t_6$ in favour of $t_1$, but
not clever enough to drop $t_3$ in favour of $t_1$.

The linear time algorithm discussed above makes all enabled actions in
$\redf(s)$ its key actions.
Some other stubborn set methods than the deadlock detection method exploit this
(e.g.,~\cite{Valmari2017b}), so it is a good idea to make it show in the
conditions.
Therefore, an alternative to \textbf{D2w} has been defined that says that all
enabled actions in $\redf(s)$ must be key actions.
To avoid choosing $\redf(s) = \emptyset$ when there are enabled actions, yet
another condition \textbf{D0} is introduced.

\begin{description}[leftmargin=!,labelwidth=\widthof{\bfseries D2:}]
	\item[\textbf{D0}] If $\enabled(s) \neq \emptyset$, then $\redf(s) \cap 
	\enabled(s) \neq \emptyset$.
\item[\textbf{D2}] Every enabled action in $\redf(s)$ is its key action in
$s$.
\end{description}

Clearly \textbf{D0} and \textbf{D2} together imply \textbf{D2w}, and
\textbf{D2w} implies \textbf{D0}.
Methods that build on \textbf{D2} are called \emph{strong} stubborn set
methods, while those only assuming \textbf{D2w} are \emph{weak}.


Please remember that the set $\Act$ of all actions is partitioned to the set
$\Inv \subseteq \Act$ of invisible actions and the set $\Act \setminus \Inv$
of visible actions.
We recall how \textbf{D1} was used in the proof of Theorem~\ref{thm:deadlock}.
The full LSTS contains the path $s_0 \transition{a_1 \ldots a_i} s_i$ where
$s_0 \in S_r$, $a_i \in \redf(s)$ and $a_j \notin \redf(s)$ for $1 \leq j <
i$.
\textbf{D1} implies the existence of $s'_0$ and the path $s_0 \transition{a_i}
s'_0 \transition{a_1 \ldots a_{i-1}} s_i$ such that $s_0 \redtransition{a_i}
s'_0$.
This pattern repeats in many proofs in the stubborn set theory.
The following condition guarantees that when using the pattern, the projection
of the action sequence on the visible actions does not change.
\begin{description}[leftmargin=!,labelwidth=\widthof{\bfseries V:}]
	\item[\textbf{V}] If $\redf(s)$ contains an enabled visible action, then it 
	contains all visible actions.
\end{description}
\begin{lem}\label{lem:D1_vis}
Assume that \textbf{D1} yields $\rho = s \transition{a_i a_1 \ldots a_{i-1}} s'$
from $\pi = s \transition{a_1 \ldots a_{i-1} a_i} s'$.
If \textbf{V} holds, then $\vis(\pi) = \vis(\rho)$.
\end{lem}
\begin{proof}
If $a_i$ is invisible, then $\vis(a_i a_1 \ldots a_{i-1}) = \vis(a_1 \ldots
a_{i-1}) = \vis(a_1 \ldots a_{i-1} a_i)$.
From now on assume that $a_i$ is visible.
Because \textbf{D1} only applies to $s \transition{a_1 \ldots a_{i-1} a_i} s'$ 
when $a_i \in \redf(s)$, $\redf(s)$ contains an enabled visible action.
By \textbf{V}, $\redf(s)$ contains all visible actions.
Because none of $a_1$, \ldots, $a_{i-1}$ is in $\redf(s)$, they must be
invisible.
So $\vis(a_i a_1 \ldots a_{i-1}) = a_i = \vis(a_1 \ldots a_{i-1} a_i)$.
\end{proof}

The application of Lemma~\ref{lem:D1_vis} to the proof of
Theorem~\ref{thm:deadlock} yields the following.
\begin{thm}\label{thm:dl_vis}
Assume that each $\redf(s)$ obeys \textbf{D1}, \textbf{D2w} and \textbf{V}.
If $s \in S_r$ and $s_n$ is a deadlock in $\TS$, then for all paths $\pi = s 
\transition{a_1 \ldots a_n} s_n$, there is a path $\rho = s 
\redtransition{b_1 \ldots b_n} s_n$ such that $\vis(\pi) = \vis(\rho)$.
\end{thm}

This theorem almost gives item (1) of the list in
Section~\ref{sec:weak_stutter}.
What is missing is that the path $s \redtransition{b_1 \ldots b_n} s_n$ is
deadlocking.
It is, because ${\rededgerel} \subseteq {\edgerel}$, so $\enabled_{\TS_r}(s_n)
\subseteq \enabled_{\TS}(s_n) = \emptyset$.
Item (2) is next to trivial.
If $s \redtransition{b_1 \ldots b_n} s_n$ is deadlocking, then $s
\transition{b_1 \ldots b_n} s_n$, and $s_n$ is a deadlock by \textbf{D2w}.

We now have sufficient background on stubborn sets to illustrate the
inconsistent labelling problem, but insufficient background to illustrate it
in a street-credible context.
Therefore, we continue and develop the LTL$_{-X}$-preserving stubborn set
method in full, and postpone the illustration of the inconsistent labelling
problem to Section~\ref{sec:counter_example}.

\subsection{Infinite Paths}

In the remainder of this paper, we will assume that the reduced LSTS is
finite.
This assumption is needed to make the next lemma hold in the presence of
non-deterministic actions.
It will be used in proving that each infinite path in $\TS$ maps to an
infinite path in $\TS_r$ with certain properties.

\begin{lem}\label{lem:inf_key}
Assume that $\redf(s_0)$ obeys \textbf{D1}, \textbf{D2w} and \textbf{V}, and
the reduced LSTS is finite.
Let $\pi = s_0 \transition{a_1} s_1 \transition{a_2} \ldots$ be any path where 
none of the $a_i$ is in $\redf(s_0)$.
Then there is a path $\rho = s_0 \redtransition{\keyact} s'_0 \transition{a_1} 
s'_1 \transition{a_2} \ldots$ for some action $\keyact$.
If, furthermore, $\keyact$ is visible, then all the $a_i$ are invisible.
\end{lem}
\begin{proof}
We use K\"onig's Lemma type of reasoning~\cite{Konig1927}.
Let $\keyact \in \redf(s_0)$ be some key action for $\redf(s_0)$.
Its existence follows from \textbf{D2w}.
By the key action property there are $s'_{0,0}$, $s'_{1,1}$, \ldots\ such that
$s_i \transition{\keyact} s'_{i,i}$.
If $\keyact$ is visible, then \textbf{V} and $a_n \notin \redf(s_0)$ for $n
\geq 1$ imply that $a_1$, $a_2$, \ldots\ are invisible.
By \textbf{D1}, for each $i$ and each $0 \leq j < i$ there are $s'_{i,j}$ such
that $s_0 \redtransition{\keyact} s'_{i,0} \transition{a_1} s'_{i,1}
\transition{a_2} \ldots \transition{a_i} s'_{i,i}$.
See Figure~\ref{fig:exist_infinite_path}.
We prove by induction that for every $k$, there is $s'_k$ such that $s_0
\redtransition{\keyact} s'_0$ (for $k = 0$) or $s'_{k-1} \transition{a_k} s'_k$
(for $k > 0$), and $s'_k = s'_{i,k}$ for infinitely many values of $i$.

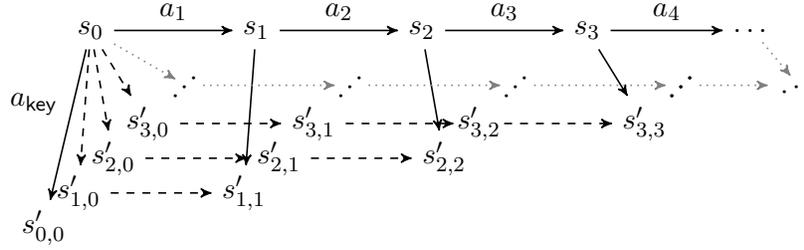
\begin{figure}
	\centering
	\begin{tikzpicture}[->,>=stealth',shorten >=0pt,auto,node 
	distance=2.0cm,semithick,
	every place/.style={draw,minimum size=5mm}]
	
	\def\x{2.2}
	\def\y{2.0}
	\def\z{0.6}
	
	\node (s0) at (0*\x,\y,0) {$s_0$};
	\node (s1) at (1*\x,\y,0) {$s_1$};
	\node (s2) at (2*\x,\y,0) {$s_2$};
	\node (s3) at (3*\x,\y,0) {$s_3$};
	\node (s4) at (4*\x,\y,0) {$\dots$};
	
	\node (s00) at (0*\x,0,2.7*\z)  {$s'_{0,0}$};
	\node (s10) at (0*\x,0,0.7*\z)  {$s'_{1,0}$};
	\node (s11) at (1*\x,0,0.7*\z)  {$s'_{1,1}$};
	\node (s20) at (0*\x,0,-1.3*\z) {$s'_{2,0}$};
	\node (s21) at (1*\x,0,-1.3*\z) {$s'_{2,1}$};
	\node (s22) at (2*\x,0,-1.3*\z) {$s'_{2,2}$};
	\node (s30) at (0*\x,0,-3.3*\z) {$s'_{3,0}$};
	\node (s31) at (1*\x,0,-3.3*\z) {$s'_{3,1}$};
	\node[inner sep=1pt] (s32) at (2*\x,0,-3.3*\z) {$s'_{3,2}$};
	\node (s33) at (3*\x,0,-3.3*\z) {$s'_{3,3}$};
	\node[rotate=45] (s40) at (0*\x,0,-5.5*\z) {$\dots$};
	\node[rotate=45] (s41) at (1*\x,0,-5.5*\z) {$\dots$};
	\node[rotate=45] (s42) at (2*\x,0,-5.5*\z) {$\dots$};
	\node[rotate=45] (s43) at (3*\x,0,-5.5*\z) {$\dots$};
	\node[rotate=45] (s43) at (3*\x,0,-5.5*\z) {$\dots$};
	\node[rotate=23] (s44) at (3.7*\x,0,-5.5*\z) {$\dots$};
	
	\path
	(s0) edge node {$\act_1$} (s1)
	(s1) edge node {$\act_2$} (s2)
	(s2) edge node {$\act_3$} (s3)
	(s3) edge node {$\act_4$} (s4)
	(s0) edge['] node {$\keyact$} (s00)
	(s1) edge (s11)
	(s2) edge (s22)
	(s3) edge (s33);
	
	\path[dashed]
	(s0) edge (s10)
	(s0) edge (s20)
	(s0) edge (s30)
	(s10) edge (s11)
	(s20) edge (s21)
	(s21) edge (s22)
	(s30) edge (s31)
	(s31) edge (s32)
	(s32) edge (s33);
	
	\path[gray,dotted]
	(s4) edge (s44)
	(s0) edge (s40)
	(s40) edge (s41)
	(s41) edge (s42)
	(s42) edge (s43)
	(s43) edge (s44);
	
	\end{tikzpicture}
	\caption{Illustration of the proof of Lemma~\ref{lem:inf_key}.
		Vertical transitions are labelled with $\keyact$; dashed transitions
		have been obtained by applying \textbf{D1}.
	}
	\label{fig:exist_infinite_path}
\end{figure}

Because there are only finitely many states, there is a state $s'_0$ that is
the same as $s'_{i,0}$ for infinitely many values of $i$.
This constitutes the base case.

To prove the induction step, we observe that all or all but one of the
infinitely many $i$ with $s'_{i,k} = s'_k$ satisfy $i > k$, and thus have an
$s'_{i,k+1}$ such that $s'_k \transition{a_{k+1}} s'_{i,k+1}$.
Infinitely many of these $s'_{i,k+1}$ are the same state, again because there
are only finitely many states.
This state qualifies as $s'_{k+1}$.
\end{proof}

\begin{lem}\label{lem:inf_prefix}
Assume that each $\redf(s)$ obeys \textbf{D1}, \textbf{D2w} and \textbf{V},
and the reduced LSTS is finite.
If $s_0 \in S_r$ and $\pi = s_0 \transition{a_1 a_2 \ldots}$, then there is a 
path $\rho = s_0 \redtransition{b_1 b_2 \ldots}$ such that $\vis(\pi)$ is a 
prefix of $\vis(\rho)$ or $\vis(\rho)$ is a prefix of $\vis(\pi)$.
\end{lem}
\begin{proof}
We use induction to prove, for each $i > 0$, the existence of $b_i$ and $s_i$
such that $s_{i-1} \redtransition{b_i} s_i$, and for each $i \geq 0$, the
existence of a path $s_i \transition{a_{i,1} a_{i,2} \ldots}$ and either
\begin{enumerate}
\item $\vis(s_0 \transition{b_1 \ldots b_i} s_i \transition{a_{i,1} a_{i,2} 
\ldots}) = \vis(\pi)$, or
\item $a_{i,1}$, $a_{i,2}$, \ldots\ are invisible and $\vis(\pi)$
is a prefix of $\vis(s_0 \transition{b_1 \ldots b_i} s_i)$.
\end{enumerate}

The base case $i = 0$ of the induction is obtained by choosing $s_0 
\transition{a_{0,1} a_{0,2} \ldots} = \pi$.
Thanks to $b_1 \ldots b_0 = \varepsilon$, (1) holds.

Regarding the induction step, if at least one of $a_{i,1}$, $a_{i,2}$, \ldots\
is in $\redf(s_i)$, then \textbf{D1} can be applied to the first such
$a_{i,j}$, yielding $s_i \redtransition{a_{i,j}} s_{i+1} \transition{a_{i,1}
\ldots a_{i,j-1} a_{i,j+1} \ldots}$.
This specifies $s_{i+1}$, and we choose $b_{i+1} = a_{i,j}$ and $a_{i+1,1}
a_{i+1,2} \ldots = a_{i,1} \ldots a_{i,j-1} a_{i,j+1} \ldots$.
We call this ``moving $a_{i,j}$ to the front''.
If (1) holds, then we apply Lemma~\ref{lem:D1_vis} and the induction hypothesis 
to deduce
\begin{align*}
&\vis(s_0 \transition{b_1 \ldots b_{i+1}} s_{i+1} \transition{a_{i+1,1} 
	a_{i+1,2} \ldots})\\
=\;\; &\vis(s_0 \transition{b_1 \ldots b_i} s_i \transition{a_{i,j}} s_{i+1} 
	\transition{a_{i,1} \ldots a_{i,j-1} a_{i,j+1} \ldots})\\
\stacksmash{(\!L\ref{lem:D1_vis}\!)}{=} &\vis(s_0 \transition{b_1 \ldots b_i} 
s_i \transition{a_{i,1} a_{i,2} 
	\ldots})\\
\stacksmash{(\!\mathit{I\!H}\!)}{=}\; &\vis(\pi)
\end{align*} 
Therefore, $i+1$ satisfies (1).
Otherwise (2) holds, implying $\{b_{i+1}, a_{i+1,1}, a_{i+1,2}, \ldots\} =
\{a_{i,1}, a_{i,2}, \ldots\} \subseteq \Inv$ and $\vis(s_0 \transition{b_1 
\ldots b_{i+1}} s_{i+1}) = \vis(s_0 \transition{b_1 \ldots b_i} s_i)$, of which 
$\vis(\pi)$ is a prefix.
Thus $i+1$ satisfies (2).

In the opposite case none of $a_{i,1}$, $a_{i,2}$, \ldots\ is in $\redf(s_i)$.
By \textbf{D2w}, $\redf(s_i)$ contains at least one key action.
To present later a further result, we choose an invisible key action if
available, and otherwise a visible one.
Lemma~\ref{lem:inf_key} yields $s_{i+1}$ such that $s_i \redtransition{\keyact}
s_{i+1} \transition{a_{i,1} a_{i,2} \ldots}$.
We choose $b_{i+1} = \keyact$ and $a_{i+1,1} a_{i+1,2} \ldots = a_{i,1}
a_{i,2} \ldots$.
We call this ``introducing a key action''.
If $\keyact \in \Inv$, then the equations
\begin{align*}
\vis(s_0 \transition{b_1 \ldots b_{i+1}} s_{i+1} \transition{a_{i+1,1} 
a_{i+1,2} \ldots}) &= \vis(s_0 \transition{b_1 \ldots b_i} s_i 
\transition{a_{i,1} a_{i,2} \ldots})\\
\vis(s_0 \transition{b_1 \ldots b_{i+1}} s_{i+1}) &= \vis(s_0 \transition{b_1 
\ldots b_i} s_i)
\end{align*}
both hold, so (1) or (2) remains valid in the step from $i$ to $i+1$.
Otherwise, $\keyact$ is visible.
Lemma~\ref{lem:inf_key} says that $a_{i,1}$, $a_{i,2}$, \ldots\ are invisible.
Then both (1) and (2) imply that $\vis(\pi)$ is a prefix of
$\vis(s_0 \transition{b_1 \ldots b_i} s_i)$, which is a (proper) prefix of 
$\vis(s_0 \transition{b_1 \ldots b_{i+1}} s_{i+1})$.
Thus also $i+1$ satisfies (2).

If (1) holds for every $i \geq 0$, then $\vis(\rho)$ is a prefix of $\vis(\pi)$.
Otherwise there is $i$ such that (2) holds.
For that and every bigger $i$, $\vis(\pi)$ is a prefix of $\vis(s_0 
\transition{b_1 \ldots b_i} s_i)$.
Therefore, $\vis(\pi)$ is a prefix of $\vis(\rho)$.
\end{proof}

The above result is a step towards item (3) of the list in
Section~\ref{sec:weak_stutter}, but not sufficient as such.
Instead, $\vis(\pi) = \vis(\rho)$ is needed.
We next add a condition, \emph{viz.}\xspace condition~\textbf{I}, guaranteeing 
that $\vis(\rho)$ is a prefix of $\vis(\pi)$.
Then we add another condition (\emph{viz.}~\textbf{L}) for the opposite 
direction.

\begin{description}[leftmargin=!,labelwidth=\widthof{\bfseries I:}]
	\item[\textbf{I}] If an invisible action is enabled, then $\redf(s)$ 
	contains an invisible key action.
\end{description}

\begin{lem}\label{lem:condition_I}
If \textbf{I} is added to the assumptions of Lemma~\ref{lem:inf_prefix}, then
$\vis(\rho)$ is a prefix of $\vis(\pi)$.
\end{lem}
\begin{proof}
Consider the proof of Lemma~\ref{lem:inf_prefix}.
By Lemma~\ref{lem:inf_key}, when none of $a_{i,1}$, $a_{i,2}$, \ldots\ is in
$r(s_i)$, then either $\keyact$ or $a_{i,1}$ is invisible.
Obviously $a_{i,1}$ is enabled in $r(s_i)$.
So \textbf{I} guarantees that there is an invisible key action.
This makes (1) remain true throughout the proof of Lemma~\ref{lem:inf_prefix},
from which the claim follows.
\end{proof}

Both \textbf{V} and \textbf{I} are easy to take into account in
$\leadsto$-based algorithms for computing strong stubborn sets.
It is much harder to ensure that $\vis(\pi)$ is a prefix of $\vis(\rho)$.
The following condition is more or less the best known.
It is usually implemented by constructing the reduced LSTS in depth-first
order so that cycles can be recognised, and using a set that contains all
visible actions as $r(s)$ in one or the other end of the edge that closes the
cycle.
\begin{description}[leftmargin=!,labelwidth=\widthof{\bfseries L:}]
	\item[\textbf{L}] For every visible action $\act$, every cycle in the 
	reduced LSTS contains a state $s$ such that $\act \in \redf(s)$.
\end{description}

\begin{lem} \label{lem:condition_L}
If \textbf{L} is added to the assumptions of Lemma~\ref{lem:inf_prefix}, then
$\vis(\pi)$ is a prefix of $\vis(\rho)$.
\end{lem}
\begin{proof}
To derive a contradiction, assume that $\vis(\pi)$ is not a prefix of 
$\vis(\rho)$.
By Lemma~\ref{lem:inf_prefix}, $\vis(\rho)$ is a proper prefix of $\vis(\pi)$.
Therefore, $\vis(\rho)$ is finite, that is, there is $i$ such that
$\vis(\rho)$ = $\vis(s_0 \transition{b_1 \ldots b_i} s_i)$.
These contradict (2) in the proof of the lemma, so (1) holds.
By it and the proper prefix property, there is $v$ such that $a_{i,v}$ is
visible.
We use the smallest such $v$.

Observe that if \textbf{D1} is applied at $s_i$ to move action $a_{i,j}$ to the 
front, where $j > v$, or \textbf{D2w} is applied, then $a_{i+1,k} = a_{i,k}$ 
for $1 \leq k \leq v$.
If the same also happens at $s_{i+1}$ then $a_{i+2,k} = a_{i,k}$ for $1 \leq k 
\leq v$, and so on, either forever or until \textbf{D1} is applied such that $j 
\leq v$, whichever comes first.
We show next that the latter comes first.

Because $S_r$ is finite, we may let $n = i + \cardinality{S_r}$.
By the pigeonhole principle, $s_i$, \ldots, $s_n$ cannot all be distinct.
So the path $s_i \transition{b_{i+1} \ldots b_{n}} s_n$ contains a cycle.
\textbf{L} implies that there is $i \leq \ell < n$ such that $a_{i,v} \in
\redf(s_\ell)$.
This guarantees that there is the smallest $h$ such that $i \leq h < i + 
\cardinality{S_r}$ and $\{a_{i,1}, \dots, a_{i,v}\} \cap \redf(s_h) \neq 
\emptyset$.
Observe that at any step $i \leq i' < h$, whether \textbf{D1} is applied to 
move $a_{i',j}$ forward, where $j > v$, or \textbf{D2w} is applied to introduce 
a key action, we have $a_{i'+1,v} = a_{i',v}$.
By \textbf{D1}, $b_{h+1}$ is one of $a_{h,1}$, \ldots, $a_{h,v}$.
So either $b_{h+1} = a_{i,v}$ or $a_{i,v} = a_{h+1,v-1}$.

Repeating the argument at most $v$ times proves that there is $i \leq h <
i + v \cardinality{S_r}$ such that $b_{h+1} = a_{i,v}$.
Because $a_{i,v}$ is visible, this contradicts $\vis(\rho)$ =
$\vis(s_0 \transition{b_1 \ldots b_i} s_i)$.
\end{proof}

We have proven the following.

\begin{thm}\label{thm:inf_vis}
Assume that each $\redf(s)$ obeys \textbf{D1}, \textbf{D2w}, \textbf{V},
\textbf{I} and \textbf{L}.
For all $s \in S_r$ and $\pi = s \transition{a_1 a_2 \ldots}$, there is a path 
$\rho = s \redtransition{b_1 b_2 \ldots}$ such that $\vis(\pi) = \vis(\rho)$.
\end{thm}

This theorem gives item (3) of the list in Section~\ref{sec:weak_stutter}.
Item (4) follows immediately from ${\rededgerel} \subseteq {\edgerel}$.
We have proven items (1) to (4) of the list in Section~\ref{sec:weak_stutter}.
Before we continue with an example of the conditions at work, we restate them 
for convenience.

\begin{description}[leftmargin=!,labelwidth=\widthof{\bfseries D2w:}]
	\item[\textbf{D0}] If $\enabled(s) \neq \emptyset$, then $\redf(s) \cap 
	\enabled(s) \neq \emptyset$.
	\item[\textbf{D1}] For all states $s_1,\dots,s_n,s'_n$ and all $\act 
	\in \redf(s)$ and $\act_1 \notin \redf(s), \dots, \act_n \notin \redf(s)$, 
	if $s \transition{\act_1} \dots \transition{\act_n} s_n \transition{\act} 
	s'_n$, then there are states $s',s'_1,\dots,s'_{n-1}$ such that $s 
	\transition{\act} s' \transition{\act_1} s'_1 \transition{\act_2} \dots 
	\transition{\act_n} s'_n$.
	\item[\textbf{D2}] Every enabled action in $\redf(s)$ is its key action in
	$s$.
	\item[\textbf{D2w}] If $\enabled(s) \neq \emptyset$, then $\redf(s)$ 
	contains a key action in $s$.
	\item[\textbf{V}] If $\redf(s)$ contains an enabled visible action, then it 
	contains all visible actions.
	\item[\textbf{I}] If an invisible action is enabled, then $\redf(s)$ 
	contains an invisible key action.
	\item[\textbf{L}] For every visible action $\act$, every cycle in the 
	reduced LSTS contains a state $s$ such that $\act \in \redf(s)$.
\end{description}

Recall that weak stubborn sets assume that conditions \textbf{D1}, 
\textbf{D2w}, \textbf{V}, \textbf{I} and \textbf{L} hold for all $\redf(s)$, 
while strong stubborn sets assume \textbf{D0}, \textbf{D1}, \textbf{D2}, 
\textbf{V}, \textbf{I} and \textbf{L} for all $\redf(s)$.

\subsection{An Example}\label{sec:exa}

Consider the Petri net and its LSTS in Figure~\ref{fig:example_por}.
We choose $\AP = \{q\}$, and $L(m) = \{q\}$ if and only if $m(p_4) > 0$
(otherwise $L(m) = \emptyset$), and illustrate this choice with grey colour on
$p_4$ and on those states where $q$ holds.
The dashed states and transitions are present in the original LSTS, but not
in the reduced version.
Other LSTSs later in this paper are visualised in a similar way.

\begin{figure}
	\centering
	\begin{tikzpicture}[->,>=stealth',shorten >=0pt,auto,node 
	distance=2.0cm,semithick,
	every place/.style={draw,minimum size=4.5mm}]
	
	\begin{scope}[label distance=-0.1cm]
	\tikzstyle{vtransition} = [fill,inner sep=0pt,minimum 
	width=1.4mm,minimum height=5mm]
	\tikzstyle{htransition} = [fill,inner sep=0pt,minimum 
	width=5mm,minimum height=1.4mm]
	
	\def\d{1.05}
	\node[place,label={above:$p_1$},tokens=1]       (p1) at (0,2*\d) {};
	\node[place,label={above:$p_2$}]                (p2) at (2*\d,2*\d) {};
	\node[place,label={above:$p_3$}]                (p3) at (4*\d,2*\d) {};
	\node[place,label={above:$p_4$},fill=lightgray] (p4) at (2*\d,\d) {};
	\node[place,label={above:$p_5$},tokens=1]       (p5) at (0,0) {};
	\node[place,label={above:$p_6$}]                (p6) at (2*\d,0) {};
	\node[place,label={right:$p_7$}]                (p7) at (4*\d,0) {};
	
	\node[vtransition,label={above:$a$}] (a) at (\d,2*\d) {};
	\node[vtransition,label={above:$w$}] (w) at (3*\d,2*\d) {};
	\node[vtransition,label={above:$b$}] (b) at (3*\d,\d) {};
	\node[htransition,label={right:$d$}] (d) at (4*\d,\d) {};
	\node[vtransition,label={above:$v$}] (v) at (\d,0) {};
	\node[vtransition,label={above:$c$}] (c) at (3*\d,0) {};
	
	\path
	(p1) edge (a) (a) edge (p2) (p2) edge (w) (w) edge (p3)
	(p5) edge (v) (v) edge (p4) (v)  edge (p6)
	(p4) edge (w) (p6) edge (b) (b)  edge (p7) (p7) edge (c)
	(c) edge (p6)
	(p3) edge[bend right=10] (d.120) (d.60) edge[bend right=10] (p3)
	(p7) edge[bend right=10] (d.310) (d.240) edge[bend right=10] (p7)
	;
	
	\end{scope}
	
	\begin{scope}[xshift=6.8cm,yshift=2.8cm]
	\tikzstyle{state}=[draw,inner sep=4pt,circle]
	\def\d{1.6}

	\node[state,label={above:$s_1$}]                 (s00) at (0,0) {};
	\node[state,label={above:$s_2$}]                 (s10) at (\d,0) {};
	\node[state,fill=lightgray,dashed]               (s01) at (0*\d,-1*\d) 
	{};
	\node[state,label={85:$s_3$},fill=lightgray]     (s11) at (1*\d,-1*\d) 
	{};
	\node[state,label={above:$s_6$}]                 (s21) at (2*\d,-1*\d) 
	{};
	\node[state,fill=lightgray,dashed]               (s02) at (0*\d,-2*\d) 
	{};
	\node[state,label={below:$s_4$},fill=lightgray]  (s12) at (1*\d,-2*\d) 
	{};
	\node[state,label={below:$s_5$}]                 (s22) at (2*\d,-2*\d) 
	{};

	\path
	(-0.5,0) edge (s00)
	(s00) edge    node {$a$} (s10)
	(s10) edge['] node {$v$} (s11)
	(s22) edge[loop right] node {$d$} (s22)
	(s11) edge[',bend right=9] node {$b$} (s12)
	(s12) edge[',bend right=9] node {$c$} (s11)
	(s21) edge[',bend right=9] node {$b$} (s22)
	(s22) edge[',bend right=9] node {$c$} (s21)
	(s12) edge node {$w$} (s22)
	;
	\path[dashed]
	(s01) edge node {$a$} (s11)
	(s02) edge node {$a$} (s12)
	(s00) edge['] node {$v$} (s01)
	(s01) edge[',bend right=9] node {$b$} (s02)
	(s02) edge[',bend right=9] node {$c$} (s01)
	(s11) edge node {$w$} (s21)
	;
	\end{scope}
	\end{tikzpicture}
	\caption{
		Example of a Petri net and its corresponding LSTS, which is reduced 
		under \textbf{D1}, \textbf{D2w}, \textbf{V}, \textbf{I} and \textbf{L}.
	}
	\label{fig:example_por}
\end{figure}
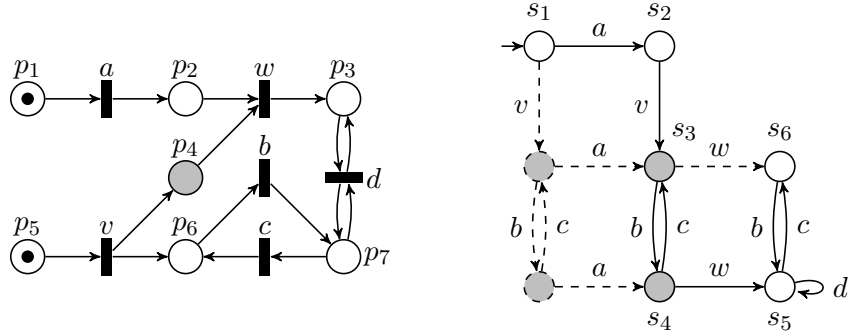

Actions $v$ and $w$ must be declared visible, because they may change the
truth value of $q$ ($v$ from false to true and $w$ in the opposite direction).
In the LSTS such events manifest themselves as transitions whose one end state
is white and the opposite end state is grey, labelled with $v$ or $w$.
Please notice that not every occurrence of a visible action must change the
truth value.
For instance, if there were initially two tokens in $p_5$, then both $\hat m
\transition{avvw}$ and $\hat m \transition{avwv}$ would be possible, the first
one inducing the label sequence $\emptyset \emptyset \{q\} \{q\} \{q\}$, and
the second $\emptyset \emptyset \{q\} \emptyset \{q\}$.

Actions $a$, $b$, $c$ and $d$ may be invisible.
In this case, we choose the set of invisible actions to be maximal, \ie, $\Inv
= \{a,b,c,d\}$.
In the initial state $s_1$, we have $\redf(s_1) = \{a\}$.
Remark that $a$ is a key action in $s_1$, since for all prefixes $\pi$ of
$v(bc)^\omega$, we have $s_1 \transition{\pi a}$.
That is, $\{a\}$ satisfies \textbf{D2w} in $s_1$.
It also satisfies \textbf{D1}, because it is easy to check that for those
$\pi$ and $s$ for which $s_1 \transition{\pi a} s$ holds, also $s_1
\transition{a \pi} s$ holds.

In states $s_3$ and $s_4$ we must have $b \in \redf(s_3)$, respectively $c \in
\redf(s_4)$, by condition \textbf{I}.
Condition \textbf{L} can be satisfied in the cycle consisting of $s_3$ and
$s_4$ by either setting $w \in \redf(s_3)$ or $w \in \redf(s_4)$; here we have
opted for the latter.
Actually, $w \in \redf(s_4)$ is also enforced by \textbf{D1}, since we have
$s_4 \transition{w d c} s_6$ and $s_4 \transition{w d} s_5$, but not $s_4
\transition{c w d} s_6$ or $s_4 \transition{d w} s_5$.
Consequently, $\{ c \}$ and $\{ c, d\}$ are not stubborn sets in $s_4$.

\section{Counter-Example}
\label{sec:counter_example}
Consider the LSTS in Figure~\ref{fig:counter_example}, which we will refer to 
as $\TS^C$.
There is only one atomic proposition $q$, which holds in the grey states and is 
false in the other states.
The initial state $\init{s}$ is marked with an incoming arrow.
First, note that this LSTS is deterministic.
The actions $\act_1$, $\act_2$ and $\act_3$ are visible and $\act$ and 
$\keyact$ are invisible.

In the initial state, we choose $\redf(\init{s}) = \{\act,\keyact\}$, which is 
a weak stubborn set by the following reasoning.
Conditions \textbf{D2w} and \textbf{I} are satisfied, since $\keyact$ is an 
invisible key action in $\init{s}$.
The path $\init{s} \transition{\act_1 \act_2}$ commutes with both $\act$ and 
$\keyact$ (and $\init{s} \transition{\act_1}$ furthermore commutes with 
$\keyact$), satisfying \textbf{D1}.
Conditions \textbf{V} and \textbf{L} are trivially true.
In all other states $s$, we choose $\redf(s) = \Act$.

As a result, we obtain a reduced LSTS $\TS^C_r$ that does not contain the 
dashed states and transitions.
The original LSTS contains the trace $\emptyset \{q\} \emptyset \emptyset 
\{q\}^\omega$, obtained by following the path with actions $\act_1 \act_2 \act 
\act_3^\omega$.
However, the reduced LSTS does not contain a stutter equivalent trace.
This is also witnessed by the LTL$_{-X}$ formula $\square (q \Rightarrow 
\square(q \lor \square \neg q))$, which holds for $\TS^C_r$, but not for 
$\TS^C$.

\begin{figure}
	\centering
	\begin{tikzpicture}[->,>=stealth',shorten >=0pt,auto,node 
	distance=2.0cm,semithick]
	\tikzstyle{state}=[draw,inner sep=4pt,circle]
	\def\x{1.7}
	\def\y{1.5}
	
	\node[state,label={above left:$\init{s}$}] (1)  at (0,\y)      {};
	\node[state,dashed,fill=lightgray]         (2)  at (\x,\y)     {};
	\node[state,dashed]                        (3)  at (2*\x,\y)   {};
	\node[state]                               (4)  at (0,0)       {};
	\node[state]                               (5)  at (\x,0)      {};
	\node[state]                               (6)  at (2*\x,0)    {};
	\node[state]                               (7)  at (0,2*\y)    {};
	\node[state,fill=lightgray]                (8)  at (\x,2*\y)   {};
	\node[state]                               (9)  at (2*\x,2*\y) {};
	\node[state,fill=lightgray]                (10) at (3*\x,0)    {};
	\path
		(-0.5,\y) edge (1)
		(1)  edge             node {$\act$}   (4)
		(4)  edge             node {$\act_1$} (5)
		(5)  edge             node {$\act_2$} (6)
		(1)  edge[']          node {$\keyact$} (7)
		(7)  edge             node {$\act_1$} (8)
		(8)  edge             node {$\act_2$} (9)
		(9)  edge[loop right] node {$\act_3$} (9)
		(6)  edge             node {$\act_3$} (10)
		(10) edge[loop right] node {$\act_3$} (10)
	;
	\path[dashed]
		(1)  edge             node {$\act_1$} (2)
		(2)  edge             node {$\act_2$} (3)
		(3)  edge             node {$\act$}   (6)
		(3)  edge[']          node {$\keyact$} (9)
		(2)  edge[']          node {$\keyact$} (8)
	;
	\end{tikzpicture}
	\caption{Counter-example showing that stubborn sets do not preserve 
	stutter-trace equivalence.
	Grey states are labelled with $\{ q \}$.
	The dashed transitions and states are not present in the reduced LSTS.}
	\label{fig:counter_example}
\end{figure}
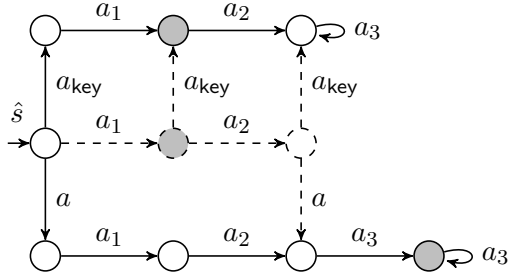

A very similar example can be used to show that strong stubborn sets suffer 
from the same problem.
Consider again the LSTS in Figure~\ref{fig:counter_example}, but assume that $a 
= \keyact$, making the LSTS no longer deterministic.
Now, $r(\init{s}) = \{\act\}$ is a strong stubborn set: \textbf{D0} is 
satisfied because $\redf(s) \cap \enabled(s) = \{a\}$ and \textbf{D2} and 
\textbf{I} are satisfied because $\act$ is an invisible key action.
Condition \textbf{D1} holds as well, since there is path $\init{s} 
\transition{\act \act_1 \act_2} s'$ (resp $\init{s} \transition{\act \act_1} 
s'$) for every path of the shape $\init{s} \transition{\act_1 \act_2 \act} s'$ 
(resp. $\init{s} \transition{\act_1 \act} s'$).
Conditions \textbf{V} and \textbf{L} are trivially true as before.
Again, the trace $\emptyset \{q\} \emptyset \emptyset \{q\}^\omega$ is not 
preserved in the reduced LSTS.
In Section~\ref{sec:deterministic_lstss}, we will see why the inconsistent 
labelling problem does not occur for deterministic systems under strong 
stubborn sets.

The core of the problem lies in the fact that condition \textbf{D1}, even when 
combined with \textbf{V}, does not enforce that the two paths it considers are 
stutter equivalent.
Consider the paths $s \transition{\act}$ and $s \transition{\act_1 \act_2 
\act}$ and assume that $\act \in \redf(s)$ and $\act_1 \notin \redf(s), \act_2 
\notin \redf(s)$.
Condition \textbf{V} ensures that at least one of the following two holds:
\begin{enumerate*}[label=\textnormal{(\roman*)}]
	\item $\act$ is invisible, or
	\item $\act_1$ and $\act_2$ are invisible
\end{enumerate*}.
Half of the possible scenarios are depicted in Figure~\ref{fig:paths_d1}; the 
other half are symmetric.
Again, the grey states (and only those states) are labelled with $\{ q \}$.

\begin{figure}
	\centering
	\begin{tikzpicture}[->,>=stealth',shorten >=0pt,auto,node 
	distance=2.0cm,semithick,scale=0.95]
	\tikzstyle{state}=[draw,inner sep=3.5pt,circle]
	\small
	\def\x{1.4}
	\def\y{1.3}
	\def\xd{2.9*\x}
	\def\yd{1.7*\y}
	
	\begin{scope}
	\node[state,fill=lightgray,label={left:$s$}]   (1)  at (0,\y)      {};
	\node[state]                                   (2)  at (\x,\y)     {};
	\node[state]                                   (3)  at (2*\x,\y)   {};
	\node[state,fill=lightgray]                    (4)  at (0,0)       {};
	\node[state]                                   (5)  at (\x,0)      {};
	\node[state,label={right:$s'$}]                (6)  at (2*\x,0)    {};
	\path
	(1)  edge node {$\act$}   (4)
	(4)  edge node {$\act_1$} (5)
	(5)  edge node {$\act_2$} (6)
	(1)  edge node {$\act_1$} (2)
	(2)  edge node {$\act_2$} (3)
	(3)  edge node {$\act$}   (6)
	;
	\end{scope}
	
	\begin{scope}[xshift=\xd cm]
	\node[state,label={left:$s$}]   (1)  at (0,\y)      {};
	\node[state,fill=lightgray]     (2)  at (\x,\y)     {};
	\node[state]                    (3)  at (2*\x,\y)   {};
	\node[state]                    (4)  at (0,0)       {};
	\node[state]                    (5)  at (\x,0)      {};
	\node[state,label={right:$s'$}] (6)  at (2*\x,0)    {};
	\path
	(1)  edge node {$\act$}   (4)
	(4)  edge node {$\act_1$} (5)
	(5)  edge node {$\act_2$} (6)
	(1)  edge node {$\act_1$} (2)
	(2)  edge node {$\act_2$} (3)
	(3)  edge node {$\act$}   (6)
	;
	\end{scope}
	
	\begin{scope}[xshift=2*\xd cm]
	\node[state,label={left:$s$}]                  (1)  at (0,\y)      {};
	\node[state,fill=lightgray]                    (2)  at (\x,\y)     {};
	\node[state,fill=lightgray]                    (3)  at (2*\x,\y)   {};
	\node[state]                                   (4)  at (0,0)       {};
	\node[state]                                   (5)  at (\x,0)      {};
	\node[state,fill=lightgray,label={right:$s'$}] (6)  at (2*\x,0)    {};
	\path
	(1)  edge node {$\act$}   (4)
	(4)  edge node {$\act_1$} (5)
	(5)  edge node {$\act_2$} (6)
	(1)  edge node {$\act_1$} (2)
	(2)  edge node {$\act_2$} (3)
	(3)  edge node {$\act$}   (6)
	;
	\end{scope}
	
	\begin{scope}[yshift=-\yd cm]
	\node[state,label={left:$s$}]   (1)  at (0,\y)      {};
	\node[state,fill=lightgray]     (2)  at (\x,\y)     {};
	\node[state]                    (3)  at (2*\x,\y)   {};
	\node[state]                    (4)  at (0,0)       {};
	\node[state,fill=lightgray]     (5)  at (\x,0)      {};
	\node[state,label={right:$s'$}] (6)  at (2*\x,0)    {};
	\path
	(1)  edge node {$\act$}   (4)
	(4)  edge node {$\act_1$} (5)
	(5)  edge node {$\act_2$} (6)
	(1)  edge node {$\act_1$} (2)
	(2)  edge node {$\act_2$} (3)
	(3)  edge node {$\act$}   (6)
	;
	\end{scope}
	
	\begin{scope}[xshift=\xd cm,yshift=-\yd cm]
	\node[state,label={left:$s$}]   (1)  at (0,\y)      {};
	\node[state]                    (2)  at (\x,\y)     {};
	\node[state]                    (3)  at (2*\x,\y)   {};
	\node[state]                    (4)  at (0,0)       {};
	\node[state,fill=lightgray]     (5)  at (\x,0)      {};
	\node[state,label={right:$s'$}] (6)  at (2*\x,0)    {};
	\path
	(1)  edge node {$\act$}   (4)
	(4)  edge node {$\act_1$} (5)
	(5)  edge node {$\act_2$} (6)
	(1)  edge node {$\act_1$} (2)
	(2)  edge node {$\act_2$} (3)
	(3)  edge node {$\act$}   (6)
	;
	\end{scope}
	
	\begin{scope}[xshift=2*\xd cm,yshift=-\yd cm]
	\node[state,label={left:$s$}]                  (1)  at (0,\y)      {};
	\node[state]                                   (2)  at (\x,\y)     {};
	\node[state,fill=lightgray]                    (3)  at (2*\x,\y)   {};
	\node[state]                                   (4)  at (0,0)       {};
	\node[state,fill=lightgray]                    (5)  at (\x,0)      {};
	\node[state,fill=lightgray,label={right:$s'$}] (6)  at (2*\x,0)    {};
	\path
	(1)  edge node {$\act$}   (4)
	(4)  edge node {$\act_1$} (5)
	(5)  edge node {$\act_2$} (6)
	(1)  edge node {$\act_1$} (2)
	(2)  edge node {$\act_2$} (3)
	(3)  edge node {$\act$}   (6)
	;
	\end{scope}
	
	\begin{scope}[yshift=-2*\yd cm]
	\node[state,label={left:$s$}]                  (1)  at (0,\y)      {};
	\node[state]                                   (2)  at (\x,\y)     {};
	\node[state,fill=lightgray]                    (3)  at (2*\x,\y)   {};
	\node[state]                                   (4)  at (0,0)       {};
	\node[state]                                   (5)  at (\x,0)      {};
	\node[state,fill=lightgray,label={right:$s'$}] (6)  at (2*\x,0)    {};
	\path
	(1)  edge node {$\act$}   (4)
	(4)  edge node {$\act_1$} (5)
	(5)  edge node {$\act_2$} (6)
	(1)  edge node {$\act_1$} (2)
	(2)  edge node {$\act_2$} (3)
	(3)  edge node {$\act$}   (6)
	;
	\end{scope}
	
	\begin{scope}[xshift=\xd cm,yshift=-2*\yd cm]
	\node[state,label={left:$s$}]    (1)  at (0,\y)      {};
	\node[state]                     (2)  at (\x,\y)     {};
	\node[state]                     (3)  at (2*\x,\y)   {};
	\node[state]                     (4)  at (0,0)       {};
	\node[state]                     (5)  at (\x,0)      {};
	\node[state,label={right:$s'$}]  (6)  at (2*\x,0)    {};
	\path
	(1)  edge node {$\act$}   (4)
	(4)  edge node {$\act_1$} (5)
	(5)  edge node {$\act_2$} (6)
	(1)  edge node {$\act_1$} (2)
	(2)  edge node {$\act_2$} (3)
	(3)  edge node {$\act$}   (6)
	;
	\end{scope}
	
	\begin{scope}[xshift=2*\xd cm,yshift=-2*\yd cm]
	\node[state,label={left:$s$}]                  (1)  at (0,\y)      {};
	\node[state]                                   (2)  at (\x,\y)     {};
	\node[state]                                   (3)  at (2*\x,\y)   {};
	\node[state,fill=lightgray]                    (4)  at (0,0)       {};
	\node[state,fill=lightgray]                    (5)  at (\x,0)      {};
	\node[state,fill=lightgray,label={right:$s'$}] (6)  at (2*\x,0)    {};
	\path
	(1)  edge node {$\act$}   (4)
	(4)  edge node {$\act_1$} (5)
	(5)  edge node {$\act_2$} (6)
	(1)  edge node {$\act_1$} (2)
	(2)  edge node {$\act_2$} (3)
	(3)  edge node {$\act$}   (6)
	;
	\end{scope}
	
	\def\xmargin{0.63}
	\def\ymargin{0.455}
	\draw[dash pattern=on 4pt off 4pt,rounded corners=8pt] 
	(\xd-\xmargin+0.05,-2*\yd-\ymargin+0.1) rectangle 
	(3*\xd-\xmargin+0.05,-1*\yd-\ymargin-0.03);
	\node at (2*\xd+1.9,-2*\yd-0.55) {$\act_1$ and $\act_2$ invisible};
	\draw[line width=0.9pt,dash pattern=on 1pt off 2pt,rounded corners=8pt] 
	(-\xmargin-0.05,-2*\yd-\ymargin-0.1) -- ++(2*\xd+0.10,0) -- ++(0,\yd+0.25) 
	-- ++(\xd+0.15,0) -- ++(0,2*\yd+0.45) -- ++(-3*\xd-0.2,0) -- cycle;
	\node at (0.5,\yd+0.32) {$\act$ invisible};
	\draw[rounded corners=8pt] (\xd-\xmargin+0.05,-1*\yd-\ymargin+0.15) 
	rectangle (2*\xd-\xmargin,\yd-\ymargin-0.03);
	\node at (\xd+1.3,\yd-0.3) {inconsistent labelling};
	\end{tikzpicture}
	\caption{Nine possible scenarios when $\act \in \redf(s)$ and $\act_1 
	\notin \redf(s), \act_2 \notin \redf(s)$, according to conditions 
	\textbf{D1} and \textbf{V}.
	The dotted and dashed lines indicate when $\act$ or $\act_1,\act_2$ are 
	invisible, respectively.
	}
	\label{fig:paths_d1}
\end{figure}
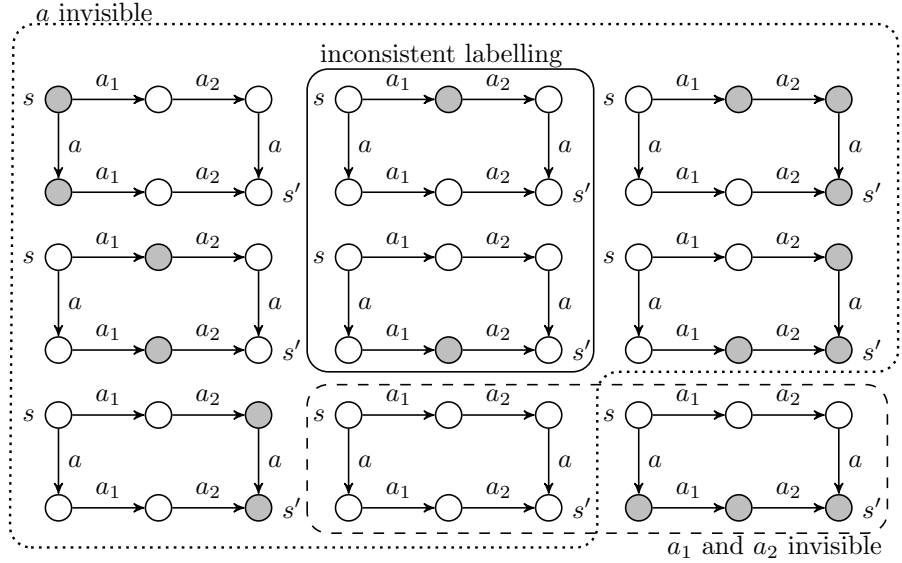

The two cases delimited with a solid line are problematic.
In both LSTSs, the paths $s \transition{\act_1 \act_2 \act} s'$ and $s 
\transition{\act \act_1 \act_2} s'$ are weakly equivalent, since $a$ is 
invisible.
However, they are not stutter equivalent, and therefore these LSTSs are not 
labelled consistently.
The topmost of these two LSTSs forms the core of the counter-example $\TS^C$, 
with the rest of $\TS^C$ serving to satisfy condition \textbf{D2}/\textbf{D2w}.

\section{Strengthening Condition D1}
\label{sec:strengthen_d1}
To fix the issue with inconsistent labelling, we propose to strengthen 
condition \textbf{D1} as follows\footnote{
	Based on a comment by one of the journal's reviewers, we noticed that 
	condition \textbf{D1'} can be weakened further, by changing the last 
	sentence to: ``Furthermore, if none of $\act_1,\dots,\act_n$ is visible, 
	then $s_i \transition{\act} s'_i$ for every $1 \leq i < n$.''
	This weakening additionally allows a reduction in the bottom-middle LSTS of 
	Figure~\ref{fig:paths_d1}, although this is hard to exploit in practice 
	(see Section~\ref{sec:implementation}).
	However, given the nature of this study, we chose to not make any 
	last-minute changes to avoid making new mistakes.
	This choice was further motivated by the following remark by another 
	reviewer (for which we are grateful): ``I really carefully checked all the 
	results and proofs and can accept the arguments and conclusions.''
}.

\begin{description}[leftmargin=!,labelwidth=\widthof{\bfseries D1':}]
	\item[\textbf{D1'}] For all states $s_1,\dots,s_n,s'_n$ and all $\act \in 
	\redf(s)$ and $\act_1 \notin \redf(s), \dots, \act_n \notin \redf(s)$, if 
	$s \transition{\act_1} s_1 \transition{\act_2} \dots \transition{\act_n} 
	s_n \transition{\act} s'_n$, then there are states $s',s'_1,\dots,s'_{n-1}$ 
	such that $s \transition{\act} s' \transition{\act_1} s'_1 
	\transition{\act_2} \dots \transition{\act_n} s'_n$.
	Furthermore, if $\act$ is invisible, then $s_i \transition{\act} s'_i$ for 
	every $1 \leq i < n$.
\end{description}

On top of what is stated in \textbf{D1}, the new condition \textbf{D1'} 
requires the presence of intermediate vertical transitions $s_i 
\transition{\act} s'_i$ whenever $\act$ is invisible.
In this case, \textbf{V} implies that $\act$ is invisible and, consequently, 
the presence of transitions $s_i \transition{\act} s'_i$ implies $L(s_i) = 
L(s'_i)$.
Thus, condition \textbf{D1'} provides a form of \emph{local} consistent 
labelling.
Hence, the problematic cases of Figure~\ref{fig:paths_d1} are resolved; a 
correctness proof is given below.

Condition \textbf{D1'} is very similar to condition 
\textbf{C1}~\cite{Gerth1999}, which is common in the context of ample sets.
However, \textbf{C1} requires that action $\act$ is \emph{globally} independent 
of each of the actions $\act_1,\dots,\act_n$, while \textbf{D1'} merely 
requires a kind of \emph{local} independence.
Persistent sets~\cite{Godefroid1996} also rely on a condition similar to 
\textbf{D1'}, and require local independence.
Thus, under ample sets and persistent sets, the vertical transitions $s_i 
\transition{\act} s'_i$ are always present, and hence they do not suffer from 
the inconsistent labelling problem.

\subsection{Correctness}
\label{sec:correctness}
To show that \textbf{D1'} indeed resolves the inconsistent labelling problem, 
we amend the lemmata and proofs of Section~\ref{sec:stubborn_sets}.
The core of the revised argument lies in a new version of 
Lemma~\ref{lem:D1_vis} that relates the state labels of the two paths 
considered by \textbf{D1'}.

\begin{lem}\label{lem:D1'_labels}
	Assume that \textbf{D1'} yields $\rho = s_0 \transition{a} s'_0 
	\transition{a_1} s'_1 \transition{a_2} \dots \transition{a_n} s'_n$	from 
	$\pi = s_0 \transition{a_1} s_1 \transition{a_2} \dots \transition{a_n} s_n 
	\transition{a} s'_n$.
	If \textbf{V} holds, then $\pi \stuteq \rho$.
\end{lem}
\begin{proof}
	If $a$ is invisible, then \textbf{D1'} enforces that $s_i 
	\transition{\act} s'_i$ for every $1 \leq i < n$.
	Thus, we have $L(s_i) = L(s'_i)$ for $1 \leq i \leq n$ and $\pi \stuteq 
	\rho$ follows.
	From now on assume that $a$ is visible.
	Because \textbf{D1'} only applies if $a \in \redf(s)$, $\redf(s)$ contains 
	an enabled visible action.
	By \textbf{V}, $\redf(s)$ contains all visible actions.
	Because none of $a_1$, \ldots, $a_n$ is in $\redf(s)$, they must be
	invisible and we have $L(s_0) = L(s_1) = \ldots = (L_n)$ and $L(s'_0) = 
	L(s'_1) = \ldots = L(s'_n)$.
	So the traces of $\pi$ and $\rho$ are $L(s_0)^{n+1} L(s'_n)$ and $L(s_0) 
	L(s'_n)^{n+1}$, respectively.
	We conclude that $\pi \stuteq \rho$.
\end{proof}

We use the same reasoning to derive the existence of a transition $s_k 
\transition{\keyact} s'_k$ for every $k > 0$ in the proof of 
Lemma~\ref{lem:inf_key}, which yields the stronger result that, if $\keyact$ is 
invisible, $\nostut(\pi) = \nostut(\rho)$.
The other lemmata are changed by replacing every occurrence of $\vis$ by 
$\nostut$.
Furthermore, in the proof of Lemma~\ref{lem:condition_L}, we reason about a 
visible action $a_{i,v}$ that actually changes the state labelling.
This results in the following two theorems that replace 
Theorems~\ref{thm:dl_vis} and~\ref{thm:inf_vis} respectively.

\begin{thm}\label{thm:dl_nostut}
	Assume that each $\redf(s)$ obeys \textbf{D1'}, \textbf{D2w} and \textbf{V}.
	If $s \in S_r$, $s_n$ is a deadlock in $\TS$, then for all paths $\pi = s 
	\transition{a_1 \ldots a_n} s_n$, there is a path $\rho = s 
	\redtransition{b_1 \ldots b_n} s_n$ such that $\nostut(\pi) = 
	\nostut(\rho)$.
\end{thm}
\begin{thm}\label{thm:inf_nostut}
	Assume that each $\redf(s)$ obeys \textbf{D1'}, \textbf{D2w}, \textbf{V},
	\textbf{I} and \textbf{L}.
	For all $s \in S_r$ and $\pi = s \transition{a_1 a_2 \ldots}$, there is a 
	path 
	$\rho = s \redtransition{b_1 b_2 \ldots}$ such that $\nostut(\pi) = 
	\nostut(\rho)$.
\end{thm}

With $\mathnormal{\rededgerel} \subseteq \mathnormal{\edgerel}$, it follows 
immediately that the replacement of condition \textbf{D1} by \textbf{D1'} is 
sufficient to ensure the reduced transition system $\TS_\redf$ is stutter-trace 
equivalent to the original transition system $\TS$.
Thus, the problem with Theorem~\ref{thm:d1_preserve_stutter_trace_equivalence} 
is resolved.

\subsection{Implementation}
\label{sec:implementation}
As discussed in Section~\ref{sec:stub_deadlock}, most, if not all, 
implementations of stubborn sets approximate \textbf{D1} based on a binary 
relation $\leadsto$ on actions.
This relation may even (partly) depend on the current state $s$, in which case 
we write $\leadsto_s$, and it should be such that condition \textbf{D1} is 
satisfied whenever $\act \in \redf(s)$ and $\act \leadsto_s \act'$ together 
imply $\act' \in \redf(s)$.
A set satisfying \textbf{D0}, \textbf{D1}, \textbf{D2}, \textbf{V} and 
\textbf{I} or \textbf{D1}, \textbf{D2w}, \textbf{V} and \textbf{I} can be found 
by searching for a suitable \emph{strongly connected component} in the graph 
$(\Act,\leadsto_s)$.
Condition \textbf{L} is dealt with by other techniques.

Practical implementations construct $\leadsto_s$ by analysing how any two 
actions $\act$ and $\act'$ interact.
If $a$ is enabled, the simplest (but not necessarily the best possible) 
strategy is to make $a \leadsto_s a'$ if and only if $a$ and $a'$ access at 
least one place (in the case of Petri nets) or variable (in the more general 
case) in common.
This can be relaxed, for instance, by not considering commutative accesses, 
such as writing to and reading from a FIFO buffer.
As a result, $\leadsto_s$ can only detect reduction opportunities in 
(sub)graphs of the shape
\begin{center}
\begin{tikzpicture}[->,>=stealth',shorten >=0pt,auto,node 
distance=2.0cm,semithick,scale=0.95]
\def\d{1.4}
\node (s)     at (0,\d)        {$s$};
\node (s1)    at (\d,\d)       {$s_1$};
\node (d)     at (1.65*\d,\d)  {$\dots$};
\node (sn-1)  at (2.5*\d,\d)   {$s_{n-1}$};
\node (sn)    at (3.5*\d,\d)   {$s_n$};
\node (s')    at (0,0)         {$s'$};
\node (s'1)   at (\d,0)        {$s'_1$};
\node (d')    at (1.65*\d,0)   {$\dots$};
\node (s'n-1) at (2.5*\d,0)    {$s'_{n-1}$};
\node (s'n)   at (3.5*\d,0)    {$s'_n$};
\path
(s)     edge node {$\act_1$} (s1)
(d)     edge              (sn-1)
(sn-1)  edge node {$\act_n$} (sn)
(sn)    edge node {$\act$}   (s'n);
\path[-]
(s1)    edge              (d);
\path
(s')    edge node {$\act_1$} (s'1)
(d')    edge              (s'n-1)
(s'n-1) edge node {$\act_n$} (s'n)
(s)     edge node {$\act$}   (s');
\path[-]
(s'1)   edge              (d');
\path
(s1)    edge node {$\act$}   (s'1)
(sn-1)    edge node {$\act$}   (s'n-1)
;
\end{tikzpicture}
\end{center}
where $\act \in \redf(s)$ and $\act_1 \notin \redf(s),\dots,\act_n \notin 
\redf(s)$.
The presence of the vertical $\act$ transitions in $s_1,\dots,s_{n-1}$ implies 
that \textbf{D1'} is also satisfied by such implementations.

\subsection{Deterministic LSTSs}
\label{sec:deterministic_lstss}
As already noted in Section~\ref{sec:counter_example}, strong stubborn sets for 
deterministic systems do not suffer from the inconsistent labelling problem.
The following lemma, which also appeared as~\cite[Lemma 4.2]{Valmari2017b}, 
shows why.

\begin{lem}
	\label{lmm:det_lsts_d1'_implied}
	For deterministic LSTSs, conditions \textbf{D1} and \textbf{D2} together 
	imply \textbf{D1'}.
\end{lem}
\begin{proof}
	Let \TS be a deterministic LSTS, $\pi = s_0 \transition{\act_1} s_1 
	\transition{\act_2} \dots \transition{\act_n} s_n \transition{\act} s'_n$ a 
	path in \TS and $\redf$ a reduction function that satisfies \textbf{D1} and 
	\textbf{D2}.
	Furthermore, assume that $\act \in \redf(s_0)$ and $\act_1 \notin 
	\redf(s_0), \dots, \act_n \notin \redf(s_0)$.
	By applying \textbf{D1}, we obtain the path $\pi' = s_0 \transition{\act} 
	s'_0 \transition{\act_1} \dots \transition{\act_n} s'_n$, which satisfies 
	the first part of condition \textbf{D1'}.
	With \textbf{D2}, we have $s_i \transition{a} s^i_i$ for every $1 \leq i 
	\leq n$.
	Then, we can also apply \textbf{D1} to every path $s_0 \transition{\act_1} 
	\dots \transition{\act_i} s_i \transition{a} s^i_i$ to obtain, for all $1 
	\leq i \leq n$, paths $\pi_i = s_0 \transition{\act} s^i_0 
	\transition{\act_1} s^i_1 \transition{\act_2} \dots \transition{\act_i} 
	s^i_i$.
	Since \TS is deterministic, every path $\pi_i$ must coincide with a prefix 
	of $\pi'$.
	We conclude that $s^i_i = s'_i$ and so the requirement that $s_i 
	\transition{\act} s'_i$ for every $1 \leq i \leq n$ is also satisfied.
\end{proof}

\section{Safe Logics}
\label{sec:safe_formalisms}
In this section, we will identify two logics, \viz reachability and CTL$_{-X}$, 
which are not affected by the inconsistent labelling problem.
This is either due to their limited expressivity or the additional POR 
conditions that are required on top of the conditions we have introduced so far.

\subsection{Reachability properties}
Although the counter-example of Section~\ref{sec:counter_example} shows that 
stutter-trace equivalence is in general not preserved by stubborn sets, some 
fragments of LTL$_{-X}$ are preserved.
One such class of properties is reachability properties, which are of the shape 
$\square f$ or $\Diamond f$, where $f$ is a formula not containing temporal 
operators.

\begin{thm}
	\label{thm:preserve_reachability}
	Let $\TS$ be an LSTS, $\redf$ a reduction function that satisfies either 
	\textbf{D0}, \textbf{D1}, \textbf{D2}, \textbf{V} and \textbf{L} or 
	\textbf{D1}, \textbf{D2w}, \textbf{V} and \textbf{L} and $\TS_r$ the 
	reduced LSTS.
	For all possible labellings $l \subseteq \AP$, $\TS$ contains an initial 
	path to a state $s$ such that $L(s) = l$ iff $\TS_r$ contains an initial 
	path to a state $s'$ such that $L(s') = l$.
\end{thm}
\begin{proof}
	The ``if'' case is trivial, since $\TS_r$ is a subgraph of $\TS$.
	For the ``only if'' case, we reason as follows.
	Let $\TS = (S, \edgerel, \init{s}, L)$ be an LSTS and $\pi = s_0 
	\transition{\act_1} \dots \transition{\act_n} s_n$ an initial path, \ie, 
	$s_0 = \init{s}$.
	We mimic this path by repeatedly taking some enabled action $\act$ that is 
	in the stubborn set, according to the following schema.
	Below, we assume the path to be mimicked contains at least one visible 
	action.
	Otherwise, its first state would have the same labelling as $s_n$.
	\begin{enumerate}
		\item If there is an $i$ such that $\act_i \in \redf(s_0)$, we consider 
		the smallest such $i$, \ie, $\act_1 \notin \redf(s_0), \dots, 
		\act_{i-1} \notin \redf(s_0)$.
		Then, we can shift $\act_i$ forward by \textbf{D1}, move towards $s_n$ 
		along $s_0 \transition{\act_i} s'_0$ and continue by mimicking $s'_0 
		\transition{\act_1} \dots \transition{\act_{i-1}} s_i 
		\transition{\act_{i+1}} \dots \transition{\act_n} s_n$.
		\item If all of $\act_1 \notin \redf(s_0), \dots, \act_n \notin 
		\redf(s_0)$, then, by \textbf{D0} and \textbf{D2} or by \textbf{D2w}, 
		there is a key action $\keyact$ in $s_0$.
		By the definition of key actions and \textbf{D1}, $\keyact$ leads to a 
		state $s'_0$ from which we can continue mimicking the path $s'_0 
		\transition{\act_1} s'_1 \transition{\act_2} \dots \transition{\act_n} 
		s'_n$.
		Note that $L(s_n) = L(s'_n)$, since $\keyact$ is invisible by condition 
		\textbf{V}.
	\end{enumerate}
	The second case cannot be repeated infinitely often, due to condition 
	\textbf{L}.
	Hence, after a finite number of steps, we reach a state $s'_n$ with 
	$L(s'_n) = L(s_n)$.
\end{proof}

We remark that more efficient mechanisms for reachability checking under POR 
have been proposed, such as condition \textbf{S}~\cite{Valmari2017a}, which can 
replace \textbf{L}, or conditions based on \emph{up-sets}~\cite{Schmidt2000}.
Another observation is that model checking of LTL$_{-X}$ properties can be 
reduced to reachability checking by computing the cross-product of a B\"uchi 
automaton and an LSTS~\cite{BaierKatoen-PMC}, in the process resolving the 
inconsistent labelling problem.
Peled~\cite{Peled1996} shows how this approach can be combined with POR, but 
please note the correctness issues detailed in~\cite{Siegel2019}.

\subsection{Deterministic LSTSs and CTL$_{-X}$ Model Checking}
\label{sec:ctl-x}
In this section, we consider the inconsistent labelling problem in the setting 
of CTL$_{-X}$ model checking.
When applying stubborn sets in that context, stronger conditions are required 
to preserve the branching structure that CTL$_{-X}$ reasons about.
Namely, the original LSTS must be deterministic and one more condition needs to 
be added~\cite{Gerth1999}:
\begin{description}[leftmargin=!,labelwidth=\widthof{\bfseries C4}]
	\item[\textbf{C4}] Either $\redf(s) = \Act$ or $\redf(s) \cap \enabled(s) = 
	\{\act\}$ for some $\act \in \Act$.
\end{description}
We slightly changed its original formulation to match the setting of stubborn 
sets.
A weaker condition, called \textbf{\"A8}, which does not require determinism of 
the whole LSTS is proposed in~\cite{Valmari1997}.
With \textbf{C4}, strong and weak stubborn sets collapse, as shown by the 
following lemma.

\begin{lem}
	\label{lmm:c4_implies_d0_d2}
	Conditions \textbf{D2w} and \textbf{C4} together imply \textbf{D0} and 
	\textbf{D2}.
\end{lem}
\begin{proof}
	Let \TS be an LSTS, $s$ a state and $\redf$ a reduction function that 
	satisfies \textbf{D2w} and \textbf{C4}.
	Condition \textbf{D0} is trivially implied by \textbf{C4}.
	Using \textbf{C4}, we distinguish two cases: either $\redf(s)$ contains 
	precisely one enabled action $\act$, or $\redf(s) = \Act$.
	In the former case, this single action $\act$ must be a key action, 
	according to \textbf{D2w}.
	Hence, \textbf{D2}, which requires that all enabled actions in $\redf(s)$ 
	are key actions, is satisfied.
	Otherwise, if $\redf(s) = \Act$, we consider an arbitrary action $\act$ 
	that satisfies \textbf{D2}'s precondition that $s \transition{\act}$.
	Given a path $s \transition{\act_1 \dots \act_n}$, the condition that 
	$\act_1 \notin \redf(s), \dots, \act_n \notin \redf(s)$ only holds if $n = 
	0$.
	We conclude that \textbf{D2}'s condition $s \transition{\act_1 \dots \act_n 
	\act}$ is satisfied by the assumption $s \transition{\act}$.
\end{proof}

It follows from Lemmata~\ref{lmm:det_lsts_d1'_implied} and 
\ref{lmm:c4_implies_d0_d2} and Theorems~\ref{thm:dl_nostut} 
and~\ref{thm:inf_nostut} that 
CTL$_{-X}$ model checking of deterministic systems with stubborn sets does not 
suffer from the inconsistent labelling problem.
The same holds for condition \textbf{\"A8}, as already shown 
in~\cite{Valmari1997}.

\section{Petri Nets}
\label{sec:petri_nets}
In this section, we discuss the impact of the inconsistent labelling problem
on Petri nets.
Contrary to Section~\ref{section:PN}, here we assume the LSTS of a Petri 
net has the set of all markings $\Markings$ as its set of states.
This does not affect the correctness of POR, as long as the set of reachable 
states $\Mreach$ is finite.
As before, we assume that the LSTS contains some labelling function $L: 
\Markings \to 2^\AP$.
More details on how a labelling function arises from a Petri net are given 
below.
Like in the Petri net examples we saw earlier, markings and structural 
transitions take over the role of states and actions respectively.
Note that the LSTS of a Petri net is deterministic.
We want to stress that all the theory in this section is specific for the 
semantics defined in Section~\ref{section:PN}.

\begin{exa}
	\label{ex:petri_net}
	Consider the Petri net with initial marking $\init{m}$ on left of 
	Figure~\ref{fig:petri_net_example}.
	Here, all arcs are weighted 1, except for the arc from $p_5$ to $\tr_3$, 
	which is weighted 2.
	Its LSTS is infinite, but the substructure reachable from $\init{m}$ is 
	depicted on the right.
	The number of tokens in each of the places $p_1,\dots,p_6$ is inscribed in 
	the nodes, the state labels (if any) are written beside the nodes.

	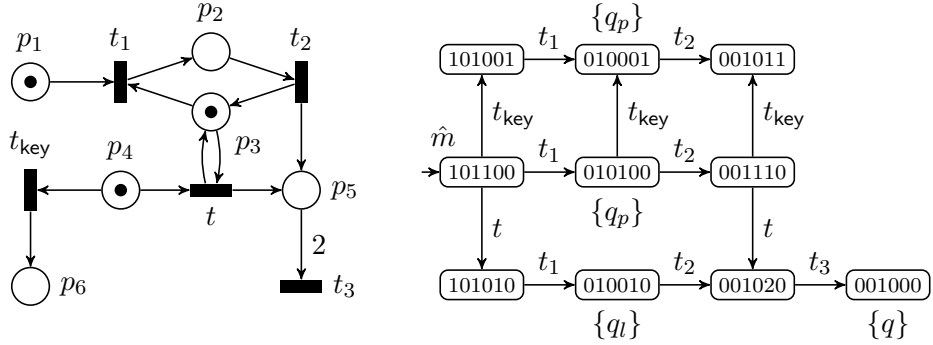
\begin{figure}
		\centering
		\begin{tikzpicture}[->,>=stealth',shorten >=0pt,auto,node 
		distance=2.0cm,semithick,
		every place/.style={draw,minimum size=5mm}]
		
		\begin{scope}[scale=0.8]
		\tikzstyle{vtransition} = [fill,inner sep=0pt,minimum 
		width=1.6mm,minimum height=5.5mm]
		\tikzstyle{htransition} = [fill,inner sep=0pt,minimum 
		width=5.5mm,minimum height=1.6mm]
		
		\def\x{1.5}
		\def\y{2}
		\node[place,label={above:$p_1$},tokens=1]       (p1) at (0,\y) {};
		\node[place,label={above:$p_2$}]                (p2) at (2*\x,1.25*\y) 
		{};
		\node[place,label={below right:$p_3$},tokens=1] (p3) at (2*\x,0.75*\y) 
		{};
		\node[place,label={above:$p_4$},tokens=1]       (p4) at (\x,0.1*\y) {};
		\node[place,label={right:$p_5$}]                (p5) at (3*\x,0.1*\y) 
		{};
		\node[place,label={right:$p_6$}]                (p6) at (0,-0.7*\y) {};
		
		\node[vtransition,label={above:$\tr_1$}]  (t1) at (\x,\y) {};
		\node[vtransition,label={above:$\tr_2$}]  (t2) at (3*\x,\y) {};
		\node[htransition,label={below:$\tr$}]    (t)  at (2*\x,0.1*\y) {};
		\node[htransition,label={right:$\tr_3$}]  (t3) at (3*\x,-0.7*\y) {};
		\node[vtransition,label={above:$\keytr$}] (tk) at (0,0.1*\y) {};
		
		\path
		(p1) edge (t1) (p3) edge (t1) (t1) edge (p2)
		(p2) edge (t2) (t2) edge (p3)
		(p3) edge[bend left=12] (t.50)  (p4) edge (t)  (t.130)  edge[bend 
		left=12] (p3)
		(t) edge (p5) (t2) edge (p5) (p5) edge node {$2$} (t3)
		(p4) edge (tk) (tk) edge (p6);
		\end{scope}
		
		\begin{scope}[xshift=6.0cm,yshift=-1.1cm]
		\tikzstyle{state}=[draw,inner sep=3pt,rectangle,rounded 
		corners=3pt,node font=\scriptsize]
		\def\x{1.8}
		\def\y{1.5}
		
		\node[state,label={above left:$\init{m}$}] (1)  at (0,\y)      {101100};
		\node[state,label={below:$\{q_p\}$}]       (2)  at (\x,\y)     {010100};
		\node[state]                               (3)  at (2*\x,\y)   {001110};
		\node[state]                               (4)  at (0,0)       {101010};
		\node[state,label={below:$\{q_l\}$}]       (5)  at (\x,0)      {010010};
		\node[state]                               (6)  at (2*\x,0)    {001020};
		\node[state]                               (7)  at (0,2*\y)    {101001};
		\node[state,label={above:$\{q_p\}$}]       (8)  at (\x,2*\y)   {010001};
		\node[state]                               (9)  at (2*\x,2*\y) {001011};
		\node[state,label={below:$\{q\}$}]         (10) at (3*\x,0)    {001000};
		\path (-0.8,\y) edge (1)
		(1)  edge    node {$\tr$}   (4)
		(4)  edge    node {$\tr_1$} (5)
		(5)  edge    node {$\tr_2$} (6)
		(1)  edge    node {$\tr_1$} (2)
		(2)  edge    node {$\tr_2$} (3)
		(3)  edge    node {$\tr$}   (6)
		(1)  edge['] node {$\keytr$} (7)
		(7)  edge    node {$\tr_1$} (8)
		(8)  edge    node {$\tr_2$} (9)
		(3)  edge['] node {$\keytr$} (9)
		(2)  edge['] node {$\keytr$} (8)
		(6)  edge    node {$\tr_3$} (10)
		;
		\end{scope}
		\end{tikzpicture}
		\caption{Example of a Petri net whose LSTS suffers from the 
		inconsistent labelling problem.}
		\label{fig:petri_net_example}
	\end{figure}

	The LSTS practically coincides with the counter-example of 
	Section~\ref{sec:counter_example}.
	Only the self-loops are missing and the state labelling, with atomic 
	propositions $q$, $q_p$ and $q_l$, differs slightly; the latter will be 
	explained later.
	For now, note that $\tr$ and $\keytr$ are invisible and that the trace 
	$\emptyset \{q_p\} \emptyset \emptyset \{q\}$, which occurs when firing 
	transitions $\tr_1 \tr_2 \tr \tr_3$ from $\init{m}$, can be lost when 
	reducing with weak stubborn sets.
	\qed
\end{exa}

In the remainder of this section, we fix a Petri net $(P, T, W, \init{m})$ and 
its LSTS $(\Markings, \edgerel, \init{m}, L)$.
Below, we consider three different types of atomic propositions.
Firstly, polynomial propositions~\cite{Bonneland2019} are of the shape 
$f(p_1,\dots,p_n) \bowtie k$ where $f$ is a polynomial over $p_1,\dots,p_n$, 
$\bowtie\, \in \{<,\leq,>,\geq,=,\neq\}$ and $k \in \mathbb{Z}$.
Such a proposition holds in a marking $m$ iff $f(m(p_1),\dots,m(p_n)) \bowtie 
k$.
A linear proposition~\cite{Liebke2019} is similar, but the function $f$ over 
places must be linear and $f(0,\dots,0) = 0$, \ie, linear propositions are of 
the shape $k_1 p_1 + \dots + k_n p_n \bowtie k$, where $k_1,\dots,k_n,k \in 
\mathbb{Z}$.
Finally, we have arbitrary propositions~\cite{Varpaaniemi2005}, whose shape is 
not restricted and which can hold in any given set of markings.

Several other types of atomic propositions can be encoded as polynomial 
propositions.
For example, $\mathit{fireable}(\tr)$~\cite{Bonneland2019,Liebke2019}, which 
holds in a marking $m$ iff $\tr$ is enabled in $m$, can be encoded as 
$\prod_{p\in P} \prod_{i = 0}^{W(p,t)-1} (p - i) \geq 1$.
The proposition $\mathit{deadlock}$, which holds in markings where no 
structural transition is enabled, does not require special treatment in the 
context of POR, since it is already preserved by \textbf{D1} and \textbf{D2w}.
The sets containing all linear and polynomial propositions are henceforward 
called $\AP_l$ and $\AP_p$, respectively.
The corresponding labelling functions are defined as $L_l(m) = L(m) \cap \AP_l$ 
and $L_p(m) = L(m) \cap \AP_p$ for all markings $m$.
Below, the two stutter equivalences $\stuteq_{L_l}$ and $\stuteq_{L_p}$ that 
follow from the new labelling functions are abbreviated $\stuteq_l$ and 
$\stuteq_p$, respectively.
Note that $\AP \supseteq \AP_p \supseteq \AP_l$ and $\mathord{\stuteq} \subseteq
\mathord{\stuteq_p} \subseteq \mathord{\stuteq_l}$.

For the purpose of introducing several variants of invisibility, we reformulate 
and generalise the definition of invisibility from 
Section~\ref{sec:preliminaries}.
Given an atomic proposition $q \in \AP$, a relation $\strel \subseteq \Markings 
\times \Markings$ is \emph{$q$-invisible} if and only if $(m, m') \in \strel$ 
implies $q \in L(m) \Leftrightarrow q \in L(m')$.
We consider a structural transition $\tr$ $q$-invisible iff its corresponding 
relation $\{(m,m') \mid m \transition{\tr} m' \}$ is $q$-invisible.
Invisibility is also lifted to sets of atomic propositions: given a set $\AP' 
\subseteq \AP$, relation $\strel$ is \emph{$\AP'$-invisible} iff it is 
$q$-invisible for all $q \in \AP'$.
If $\strel$ is $\AP$-invisible, we plainly say that $\strel$ is 
\emph{invisible}.
$\AP'$-invisibility and invisibility carry over to structural transitions.
We sometimes refer to invisibility as \emph{ordinary invisibility} for emphasis.
Note that the set of invisible structural transitions $\Inv$ is no longer an 
under-approximation, but contains exactly those structural transitions $\tr$ 
for which $m \transition{\tr} m'$ implies $L(m) = L(m')$ (cf. 
Section~\ref{sec:preliminaries}).

We are now ready to introduce three orthogonal variations on invisibility.

\begin{defi}
	Let $\strel \subseteq \Markings \times \Markings$ be a relation on markings.
	Then,
	\begin{itemize}
		\item $\strel$ is \emph{reach $q$-invisible}~\cite{Valmari2017a} iff 
		$\strel \cap (\Mreach \times \Mreach)$ is $q$-invisible; and
		\item $\strel$ is \emph{value $q$-invisible} iff
		\begin{itemize}
			\item $q = (f(p_1,\dots,p_n) \bowtie k)$ is polynomial and for all 
			pairs of markings $(m,m') \in \strel$, we have that
			$f(m(p_1),\dots,m(p_n)) = f(m'(p_1),\dots,m'(p_n))$; or
			\item $q$ is not polynomial and $\strel$ is $q$-invisible.
		\end{itemize}
	\end{itemize}
\end{defi}

Intuitively, under reach $q$-invisibility, all pairs of reachable markings 
$(m,m') \in \strel$ have to agree on the labelling of $q$.
For value invisibility, the value of the polynomial $f$ must never change 
between two markings $(m,m') \in \strel$.
Reach and value invisibility are lifted to structural transitions and sets of 
atomic propositions as before, \ie, by taking $\strel = \{ (m,m') \mid m 
\transition{\tr} m' \}$ when considering invisibility of $\tr$.

\begin{defi}
	A structural transition $\tr$ is \emph{strongly $q$-invisible} iff the set 
	$\{ (m,m') \mid \forall p \in P: m'(p) = m(p) + W(t,p) - W(p,t) \}$ is 
	$q$-invisible.
\end{defi}

Strong invisibility does not take the presence of a transition $m 
\transition{\tr} m'$ into account, and purely reasons about the effects of 
$\tr$.
Value invisibility and strong invisibility are new in the current work, 
although strong invisibility was inspired by the notion of invisibility that is 
proposed by Varpaaniemi in~\cite{Varpaaniemi2005}.
Our definition of strong invisibility weakens the conditions of Varpaaniemi.

\begin{figure} 
	\centering
	\begin{tikzpicture}[->,>=stealth',inner 
	sep=2pt,x={(1cm,0.85cm)},y=0.8cm,z={(-1cm,0.85cm)}]
	\node (is)   at (0,1,1) {$\Inv_s$};
	\node (iv)   at (1,1,0) {$\Inv_v$};
	\node (i)    at (0,1,0) {$\Inv$};
	\node (irs)  at (0,0,1) {$\Inv^r_s$};
	\node (irv)  at (1,0,0) {$\Inv^r_v$};
	\node (ir)   at (0,0,0) {$\Inv^r$};
	\node (irsv) at (1,0,1) {$\Inv^r_{sv}$};
	\node (isv)  at (1,1,1) {$\Inv_{sv}$};
	\path
	(is)  edge (i)
	(iv)  edge (i)
	(irs) edge (ir)
	(irv) edge (ir)
	(is)  edge (irs)
	(iv)  edge (irv)
	(i)   edge (ir)
	(isv) edge (irsv)
	(isv) edge (is)
	(isv) edge (iv)
	(irsv) edge (irs)
	(irsv) edge (irv)
	;
	\end{tikzpicture}
	\caption{Lattice of sets of invisible actions. Arrows represent a subset 
	relation.}
	\label{fig:invisibility_lattice}
\end{figure}
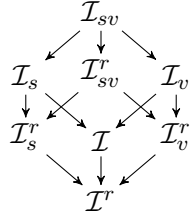%

We indicate the sets of all value, reach and strongly invisible structural 
transitions with $\Inv_v$, $\Inv^r$ and $\Inv_s$ respectively.
Since $\Inv_v \subseteq \Inv$, $\Inv_s \subseteq \Inv$ and $\Inv \subseteq 
\Inv^r$, the set of all their possible combinations forms the lattice shown in 
Figure~\ref{fig:invisibility_lattice}.
In the remainder, the weak equivalence relations that follow from each of the 
eight invisibility notions are abbreviated, \eg, $\weakeq_{\Inv^r_{sv}}$ 
becomes $\weakeq^r_{sv}$.

\begin{exa}
	\label{ex:petri_net_AP}
	Consider again the Petri net and LSTS from Example~\ref{ex:petri_net}.
	We can define $q_l$ and $q_p$ as linear and polynomial propositions, 
	respectively:
	\begin{itemize}
		\item $q_l := p_3 + p_4 + p_6 = 0$ is a linear proposition, which holds 
		when neither $p_3$, $p_4$ nor $p_6$ contains a token.
		Structural transition $\tr$ is $q_l$-invisible, because $m 
		\transition{\tr} m'$ implies that $m(p_3) = m'(p_3) \geq 1$, and thus 
		neither $m$ nor $m$ is labelled with $q_l$.
		On the other hand, $\tr$ is not value $q_l$-invisible (by the 
		transition $101100 \transition{\tr} 101010$) or strongly reach 
		$q_l$-invisible (by $010100$ and $010010$).
		However, $\keytr$ is strongly value $q_l$-invisible: it moves a token 
		from $p_4$ to $p_6$ and hence never changes the value of $p_3 + p_4 + 
		p_6$.
		\item $q_p := (1 - p_3)(1 - p_5) = 1$ is a polynomial proposition, 
		which holds in all reachable markings $m$ where $m(p_3) = m(p_5) = 0$ 
		or $m(p_3) = m(p_5) = 2$.
		Structural transition $\tr$ is reach value $q_p$-invisible, but not 
		$q_p$-invisible (by $002120 \transition{\tr} 002030$) or strongly reach 
		$q_p$ invisible.
		Strong value $q_p$-invisibility of $\keytr$ follows immediately from 
		the fact that the adjacent places of $\keytr$, \viz $p_4$ and $p_6$, do 
		not occur in the definition of $q_p$.
	\end{itemize}
	This yields the state labelling which is shown in 
	Example~\ref{ex:petri_net}.
	\qed
\end{exa}

Given a weak equivalence relation $R_\weakeq$ and a stutter equivalence 
relation $R_\stuteqsym$, we write $R_\weakeq \conslab R_\stuteqsym$ to indicate 
that $R_\weakeq$ and $R_\stuteqsym$ yield consistent labelling 
(Definition~\ref{def:consistent_labelling}).
We spend the rest of this section investigating under which notions of 
invisibility and propositions from the literature, the LSTS of a Petri net is 
labelled consistently.
More formally, we check for each weak equivalence relation $R_\weakeq$ and each 
stutter equivalence relation $R_\stuteqsym$ whether $R_\weakeq \conslab 
R_\stuteqsym$.
This tells us when existing stubborn set theory can be applied without problems.
The two lattices containing all weak and stuttering equivalence relations are 
depicted in Figure~\ref{fig:weak_stut_lattices}; each dotted arrow represents a 
consistent labelling result.
Before we continue, we first introduce an auxiliary lemma.
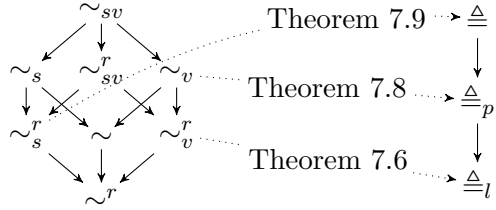
\begin{figure}[t]
	\centering
	\begin{tikzpicture}[->,>=stealth',inner sep=2pt]
	\begin{scope}[name prefix={},x={(1cm,0.85cm)},y=0.8cm,z={(-1cm,0.85cm)}]
		\node (ws)   at (0,1,1) {$\weakeq_s$};
		\node (wv)   at (1,1,0) {$\weakeq_v$};
		\node (w)    at (0,1,0) {$\weakeq$};
		\node (wrs)  at (0,0,1) {$\weakeq^r_s$};
		\node (wrv)  at (1,0,0) {$\weakeq^r_v$};
		\node (wr)   at (0,0,0) {$\weakeq^r$};
		\node (wrsv) at (1,0,1) {$\weakeq^r_{sv}$};
		\node (wsv)  at (1,1,1) {$\weakeq_{sv}$};
		\path
		(ws)  edge (w)
		(wv)  edge (w)
		(wrs) edge (wr)
		(wrv) edge (wr)
		(ws)  edge (wrs)
		(wv)  edge (wrv)
		(w)   edge (wr)
		(wsv) edge (wrsv)
		(wsv) edge (ws)
		(wsv) edge (wv)
		(wrsv) edge (wrs)
		(wrsv) edge (wrv)
		;
	\end{scope}
	\begin{scope}[xshift=5cm,yshift=0.2cm,name prefix={}]
		\def\d{1.1}
		\node (s)  at (0,2*\d) {$\stuteq$};
		\node (sp) at (0,\d) {$\stuteq_p$};
		\node (sl) at (0,0) {$\stuteq_l$};
		\path
		(s) edge (sp)
		(sp) edge (sl)
		;
	\end{scope}
	\path[dotted]
	(wrs)  edge[bend left=17] node[near end,fill=white] 
	{Theorem~\ref{thm:petri_net_labelled_consistently_arbitrary_ap}} (s)
	(wv) edge node[fill=white] 
	{Theorem~\ref{thm:petri_net_labelled_consistently_polynomial_ap}} (sp)
	(wrv) edge node[fill=white] 
	{Theorem~\ref{thm:petri_net_labelled_consistently_linear_ap}} (sl)
	;
	\end{tikzpicture}
	\caption{Two lattices containing variations of weak equivalence and stutter 
	equivalence, respectively.
	Solid arrows indicate a subset relation inside the lattice; dotted arrows 
	follow from the indicated theorems and show when the LSTS of a Petri net is 
	labelled consistently.
	}
	\label{fig:weak_stut_lattices}
\end{figure}

\begin{lem}
	\label{lmm:meta_proof_pn_consistent_labelling}
	Let $I$ be a set of invisible structural transitions and $L$ some labelling 
	function.
	If for all $\tr \in I$ and paths $\pi = m_0 \transition{\tr_1} m_1 
	\transition{\tr_2} \dots$ and $\pi' = m_0 \transition{\tr} m'_0 
	\transition{\tr_1} m'_1 \transition{\tr_2} \dots$, it holds that $\pi 
	\stuteq_L \pi'$, then $\mathord{\weakeq_I} \conslab \mathord{\stuteq_L}$.
\end{lem}
\begin{proof}
	We assume that the following holds for all paths and $\tr \in I$:
	\begin{equation}
		\tag{$\dagger$}\label{eq:pn_introduce_invis_front}
		m_0 \transition{\tr_1} m_1 \transition{\tr_2} \dots \stuteq_L m_0 
		\transition{\tr} m'_0 \transition{\tr_1} m'_1 \transition{\tr_2} \dots
	\end{equation}
	To prove $\mathord{\weakeq_I} \conslab \mathord{\stuteq_L}$, we need to 
	consider two initial paths $\pi$ and $\pi'$ such that $\pi \weakeq_I \pi'$ 
	and prove that $\pi \stuteq_L \pi'$ (see 
	Definition~\ref{def:consistent_labelling}).
	The proof proceeds by induction on the combined number of invisible 
	structural transitions (taken from $I$) in $\pi$ and $\pi'$.
	In the base case, $\pi$ and $\pi'$ contain only visible structural 
	transitions, and $\pi \weakeq_I \pi'$ implies $\pi = \pi'$ since Petri nets 
	are deterministic.
	Hence, $\pi \stuteq_L \pi'$.
	
	For the induction step, we take as hypothesis that, for all initial paths 
	$\pi$ and $\pi'$ that together contain at most $k$ invisible structural 
	transitions, $\pi \weakeq_I \pi'$ implies $\pi \stuteq_L \pi'$.
	Let $\pi$ and $\pi'$ be two arbitrary initial paths such that $\pi 
	\weakeq_I \pi'$ and the total number of invisible structural transitions 
	contained in $\pi$ and $\pi'$ is $k$.
	We consider the case where an invisible structural transition is introduced 
	in $\pi'$, the other case is symmetric.
	Let $\pi' = \sigma_1 \sigma_2$ for some $\sigma_1$ and $\sigma_2$.
	Let $\tr \in I$ be some invisible structural transition and $\pi'' = 
	\sigma_1 \tr \sigma'_2$ such that $\sigma_2$ and $\sigma'_2$ contain the 
	same sequence of structural transitions.
	Clearly, we have $\pi' \weakeq_I \pi''$.
	Here, we can apply our original assumption 
	(\ref{eq:pn_introduce_invis_front}), to conclude that $\sigma_2 \stuteq 
	t\sigma'_2$, \ie, the extra stuttering step $\tr$ thus does not affect the 
	labelling of the remainder of $\pi''$.
	Hence, we have $\pi' \stuteq_L \pi''$ and, with the induction hypothesis, 
	$\pi \stuteq_L \pi''$.
	Note that $\pi$ and $\pi''$ together contain $k+1$ invisible structural 
	transitions.
	
	In case $\pi$ and $\pi'$ together contain an infinite number of invisible 
	structural transitions, $\pi \weakeq_I \pi'$ implies $\pi \stuteq_L \pi'$ 
	follows from the fact that the same holds for all finite prefixes of $\pi$ 
	and $\pi'$ that are related by $\weakeq_I$.
\end{proof}

The following theorems each focus on a class of atomic propositions and show 
which notion of invisibility is required for the LSTS of a Petri net to be 
labelled consistently.
In the proofs, we use a function $d_\tr$, defined as $d_\tr(p) = W(\tr,p) - 
W(p,\tr)$ for all places $p$, which indicates how structural transition $\tr$ 
changes the state.
Furthermore, we also consider functions of type $P \to \mathbb{N}$ as vectors 
of type $\mathbb{N}^{\cardinality{P}}$.
This allows us to compute the pairwise addition of a marking $m$ with $d_\tr$ 
($m + d_\tr$) and to indicate that $\tr$ does not change the marking ($d_\tr = 
0$).

\begin{thm}
	\label{thm:petri_net_labelled_consistently_linear_ap}
	Under reach value invisibility, the LSTS underlying a Petri net is labelled 
	consistently for linear propositions, \ie, $\mathord{\weakeq^r_v} \conslab 
	\mathord{\stuteq_l}$.
\end{thm}
\begin{proof}
	Let $\tr \in \Inv^r_v$ be a reach value invisible structural transition 
	such that there exist reachable markings $m$ and $m'$ with $m 
	\transition{\tr} m'$.
	If such a $\tr$ does not exist, then $\weakeq^r_v$ is the reflexive 
	relation and $\mathord{\weakeq^r_v} \conslab \mathord{\stuteq_l}$ is 
	trivially satisfied.
	Otherwise, let $q := f(p_1,\dots,p_n) \bowtie k$ be a linear proposition.
	Since $\tr$ is reach value invisible and $f$ is linear, we have $f(m) = 
	f(m') = f(m + d_\tr) = f(m) + f(d_\tr)$ and thus $f(d_\tr) = 0$.
	It follows that, given two paths $\pi = m_0 \transition{\tr_1} m_1 
	\transition{\tr_2} \dots$ and $\pi' = m_0 \transition{\tr} m'_0 
	\transition{\tr_1} m'_1 \transition{\tr_2} \dots$, the addition of $\tr$ 
	does not influence $f$, since $f(m_i) = f(m_i) + f(d_\tr) = f(m_i + d_\tr) 
	= f(m'_i)$ for all $i$.
	As a consequence, $\tr$ also does not influence $q$.
	With Lemma~\ref{lmm:meta_proof_pn_consistent_labelling}, we deduce that 
	$\mathord{\weakeq^r_v} \conslab \mathord{\stuteq_l}$.
\end{proof}

Whereas in the linear case one can easily conclude that $\pi$ and $\pi'$ are 
stutter equivalent under $f$, in the polynomial case, we need to show that $f$ 
is constant under all value invisible structural transitions $\tr$, even in 
markings where $\tr$ is not enabled.
This follows from the following proposition.

\begin{prop}
	\label{prop:value_invis_f_constant}
	Let $f: \mathbb{N}^n \to \mathbb{Z}$ be a polynomial function, $a,b \in 
	\mathbb{N}^n$ two constant vectors and $c = a - b$ the difference between 
	$a$ and $b$.
	Assume that for all $x \in \mathbb{N}^n$ such that $x \geq b$, where $\geq$ 
	denotes pointwise comparison, it holds that $f(x) = f(x + c)$.
	Then, $f$ is constant in the vector $c$, \ie, $f(x) = f(x+c)$ for all $x 
	\in \mathbb{N}^n$.
\end{prop}
\begin{proof}
	Let $f$, $a$, $b$ and $c$ be as above and let $\mathbf{1} \in \mathbb{N}^n$ 
	be the vector containing only ones.
	Given some arbitrary $x \in \mathbb{N}^n$, consider the function $g_x(t) = 
	f(x + t\cdot\mathbf{1} + c) - f(x + t\cdot\mathbf{1})$.
	For sufficiently large $t$, it holds that $x + t \cdot \mathbf{1} \geq b$, 
	and it follows that $g_x(t) = 0$ for all sufficiently large $t$.
	This can only be the case if $g_x$ is the zero polynomial, \ie, $g_x(t) = 
	0$ for all $t$.
	As a special case, we conclude that $g_x(0) = f(x + c) - f(x) = 0$.
\end{proof}

The intuition behind this is that $f(x+c) - f(x)$ behaves like the directional 
derivative of $f$ with respect to $c$.
If the derivative is equal to zero in infinitely many $x$, $f$ must be constant 
in the direction of $c$.
We will apply this result in the following theorem.

\begin{thm}
	\label{thm:petri_net_labelled_consistently_polynomial_ap}
	Under value invisibility, the LSTS underlying a Petri net is labelled 
	consistently for polynomial propositions, \ie, $\mathord{\weakeq_v} 
	\conslab \mathord{\stuteq_p}$.
\end{thm}
\begin{proof}
	Let $\tr \in \Inv_v$ be a value invisible structural transition, $m$ and 
	$m'$ two markings with $m \transition{\tr} m'$, and $q := f(p_1,\dots,p_n) 
	\bowtie k$ a polynomial proposition.
	Note that infinitely many such (not necessarily reachable) markings exist 
	in $\Markings$, so we can apply 
	Proposition~\ref{prop:value_invis_f_constant} to obtain $f(m) = f(m + 
	d_\tr)$ for all markings $m$.
	It follows that, given two paths $\pi = m_0 \transition{\tr_1} m_1 
	\transition{\tr_2} \dots$ and $\pi' = m_0 \transition{\tr} m'_0 
	\transition{\tr_1} m'_1 \transition{\tr_2} \dots$, the addition of $\tr$ 
	does not alter the value of $f$, since $f(m_i) = f(m_i + d_\tr) = f(m'_i)$ 
	for all $i$.
	As a consequence, $\tr$ also does not change the labelling of $q$.
	Application of Lemma~\ref{lmm:meta_proof_pn_consistent_labelling} yields 
	$\mathord{\weakeq_v} \conslab \mathord{\stuteq_p}$.
\end{proof}

Varpaaniemi shows that the LSTS of a Petri net is labelled consistently for 
arbitrary propositions under his notion of invisibility~\cite[Lemma 
9]{Varpaaniemi2005}.
Our notion of strong visibility, and especially strong reach invisibility, is 
weaker than Varpaaniemi's invisibility, so we generalise the result to 
$\mathord{\weakeq^r_s} \conslab \mathord{\stuteq}$.
\begin{thm}
	\label{thm:petri_net_labelled_consistently_arbitrary_ap}
	Under strong reach visibility, the LSTS underlying a Petri net is labelled 
	consistently for arbitrary propositions, \ie, $\mathord{\weakeq^r_s} 
	\conslab \mathord{\stuteq}$.
\end{thm}
\begin{proof}
	Let $\tr \in \Inv^r_s$ be a strongly reach invisible structural transition 
	and $\pi = m_0 \transition{\tr_1} m_1 \transition{\tr_2} \dots$ and $\pi' = 
	m_0 \transition{\tr} m'_0 \transition{\tr_1} m'_1 \transition{\tr_2} \dots$ 
	two paths.
	Since, $m'_i = m_i + d_\tr$ for all $i$, it holds that either
	\begin{enumerate*}[label=\textnormal{(\roman*)}]
		\item $d_\tr = 0$ and $m_i = m'_i$ for all $i$; or
		\item each pair $(m_i,m'_i)$ is contained in $\{ (m,m') \mid \forall p 
		\in P: m'(p) = m(p) + W(t,p) -\linebreak[1] W(p,t) \}$, which is the 
		set that underlies strong reach invisibility of $\tr$.
	\end{enumerate*}
	In both cases, $L(m_i) = L(m'_i)$ for all $i$.
	It follows from Lemma~\ref{lmm:meta_proof_pn_consistent_labelling} that 
	$\mathord{\weakeq^r_s} \conslab \mathord{\stuteq}$.
\end{proof}

To show that the results of the above theorems cannot be strengthened, we 
provide two negative results.
\begin{thm}
	\label{thm:petri_net_not_labelled_consistently_weak_inv}
	Under ordinary invisibility, the LSTS underlying a Petri net is not 
	necessarily labelled consistently for arbitrary propositions, \ie, 
	$\mathord{\weakeq} \nconslab \mathord{\stuteq}$.
\end{thm}
\begin{proof}
	Consider the Petri net from Example~\ref{ex:petri_net} with the arbitrary 
	proposition $q_l$.
	Disregard $q_p$ for the moment.
	Structural transition $\tr$ is $q_l$-invisible, hence the paths 
	corresponding to $\tr_1 \tr_2 \tr \tr_3$ and $\tr \tr_1 \tr_2 \tr_3$ are 
	weakly equivalent under ordinary invisibility.
	However, they are not stutter equivalent.
\end{proof}
\begin{thm}
	\label{thm:petri_net_not_labelled_consistently_reach_value_inv}
	Under reach value invisibility, the LSTS underlying a Petri net is not 
	necessarily labelled consistently for polynomial propositions, \ie, 
	$\mathord{\weakeq^r_v} \nconslab \mathord{\stuteq_p}$.
\end{thm}
\begin{proof}
	Consider the Petri net from Example~\ref{ex:petri_net} with the polynomial 
	proposition $q_p := (1-p_3)(1 - p_5) = 1$ from 
	Example~\ref{ex:petri_net_AP}.
	Disregard $q_l$ in this reasoning.
	Structural transition $\tr$ is reach value $q_p$-invisible, hence the paths 
	corresponding to $\tr_1 \tr_2 \tr \tr_3$ and $\tr \tr_1 \tr_2 \tr_3$ are 
	weakly equivalent under reach value invisibility.
	However, they are not stutter equivalent for polynomial propositions.
\end{proof}

It follows from Theorems~\ref{thm:petri_net_not_labelled_consistently_weak_inv} 
and~\ref{thm:petri_net_not_labelled_consistently_reach_value_inv} and 
transitivity of $\subseteq$ that 
Theorems~\ref{thm:petri_net_labelled_consistently_linear_ap}, 
\ref{thm:petri_net_labelled_consistently_polynomial_ap} 
and~\ref{thm:petri_net_labelled_consistently_arbitrary_ap} cannot be 
strengthened further.
In terms of Figure~\ref{fig:weak_stut_lattices}, this means that the dotted 
arrows cannot be moved downward in the lattice of weak equivalences and cannot 
be moved upward in the lattice of stutter equivalences.
The implications of these findings on related work will be discussed in the 
next section.

\section{Related Work}
\label{sec:related_work}
There are many works in the literature that apply stubborn sets.
We will consider several works that aim to preserve LTL$_{-X}$ and discuss 
whether they are correct when it comes to the inconsistent labelling problem.
Furthermore, we also identify several unrelated issues.

Liebke and Wolf~\cite{Liebke2019} present an approach for efficient CTL model 
checking on Petri nets.
For some formulas, they can reduce CTL model checking to LTL model checking, 
which allows greater reductions under POR.
They rely on the incorrect LTL preservation theorem, and since they apply the 
techniques on Petri nets with ordinary invisibility, their theory is incorrect 
(Theorem~\ref{thm:petri_net_not_labelled_consistently_weak_inv}).
Similarly, the overview of stubborn set theory presented by Valmari and Hansen 
in~\cite{Valmari2017a} applies reach invisibility and does not necessarily 
preserve LTL$_{-X}$.
Varpaaniemi~\cite{Varpaaniemi2005} also applies stubborn sets to Petri nets, 
but relies on a visibility notion that is stronger than strong invisibility.
The correctness of these results is thus not affected 
(Theorem~\ref{thm:petri_net_labelled_consistently_arbitrary_ap}).

A generic implementation of weak stubborn sets for the LTSmin model checker is
proposed by Laarman \etal~\cite{Laarman2016}.
They use abstract concepts such as guards and transition groups to implement 
POR in a way that is agnostic of the input language.
The theory they present includes condition \textbf{D1}, which is too weak and 
thus incorrect, but the accompanying implementation follows the framework of 
Section~\ref{sec:implementation}, and thus it is correct by 
Theorems~\ref{thm:dl_nostut} and~\ref{thm:inf_nostut}.
The implementations proposed in~\cite{Valmari2017a,Wolf2018} are similar, 
albeit specific for Petri nets.

Several works~\cite{Gibson-Robinson2015,Hansen2014} perform action-based 
model checking and thus strive to preserve weak trace equivalence or inclusion.
As such, they do not suffer from the problems discussed here, which applies 
only to state labels.
Other recent work~\cite{dobrikov_optimising_2016} relies on ample sets, and is 
thus not affected, or only considers safety 
properties~\cite{laarman_stubborn_2018}.

Although Bene\v{s} \etal~\cite{Benes2009,Benes2011} rely on ample sets, and not
on stubborn sets, they also discuss weak trace equivalence and stutter-trace 
equivalence.
In fact, they present an equivalence relation for traces that is a combination 
of weak and stutter equivalence.
The paper includes a theorem that weak equivalence implies their new 
state/event equivalence~\cite[Theorem 6.5]{Benes2011}.
However, the counter-example in Figure~\ref{fig:ce_Benes} shows that this 
consistent labelling theorem does not hold.
Here, the action $\tau$ is invisible, and the two paths in this transition 
system are thus weakly equivalent.
However, they are not stutter equivalent, which is a special case of 
state/event equivalence.
Although the main POR correctness result~\cite[Corollary 6.6]{Benes2011} builds 
on the incorrect consistent labelling theorem, its correctness does not appear 
to be affected.
An alternative proof can be constructed based on the reasoning presented in 
Section~\ref{sec:correctness}.

\begin{figure}
	\centering
	\begin{subfigure}{0.3\textwidth}
		\centering
		\begin{tikzpicture}[->,>=stealth',shorten >=0pt,auto,node 
		distance=2.0cm,semithick]

		\tikzstyle{state} = [draw,circle]

		\node[state] (s0)                       at (0,0)      {};
		\node[state] (s1)                       at (1.2,0.6)  {};
		\node[state,label={right:$\{q\}$}] (s2) at (2.4,0.6)  {};
		\node[state] (s3)                       at (1.2,-0.6) {};

		\path (-0.5,0) edge (s0)
		(s0) edge    node {$\tau$} (s1)
		(s1) edge    node {$a$}    (s2)
		(s0) edge['] node {$a$}    (s3)
		;

		\end{tikzpicture}
		\vspace{1.15em}
		\caption{} 
		\label{fig:ce_Benes}
	\end{subfigure}
	\hspace{1cm}
	\begin{subfigure}{0.3\textwidth}
		\centering
		\begin{tikzpicture}[->,>=stealth',shorten >=0pt,auto,node 
		distance=2.0cm,semithick]
		\tikzstyle{state} = [draw,circle]
		\def\x{1.4}
		\node[state] (0) at (0,0)    {};
		\node[state,dashed] (1) at (\x,0)   {};
		\node[state,fill=lightgray] (2) at (0,-\x)  {};
		\node[state] (3) at (\x,-\x) {};
		\path (-0.5,0) edge (0)
		(0) edge['] node {$b$} (2)
		(2) edge['] node {$a$} (3)
		;
		\path[dashed]
		(0) edge node {$a$} (1)
		(1) edge node {$b$} (3)
		;
		\end{tikzpicture}
		\caption{} 
		\label{fig:ce_Bonneland}
	\end{subfigure}
	\caption{Counter-examples for theories in two related works.}
\end{figure}
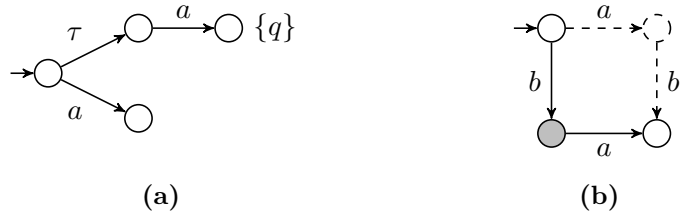

B{\o}nneland \etal~\cite{Bonneland2019} apply stubborn-set based POR to 
two-player Petri nets, and their reachability semantics expressed as a 
\emph{reachability game}.
Since their approach only concerns reachability, it is not affected by the 
inconsistent labelling problem (see Section~\ref{sec:safe_formalisms}).
Unfortunately, their POR theory is nevertheless unsound, contrary to what is 
claimed in~\cite[Theorem 17]{Bonneland2019}.
In reachability games, player 1 tries to reach one of the \emph{goal} states, 
while player 2 tries to avoid them.
B{\o}nneland \etal propose a condition \textbf{R} that guarantees that all goal 
states in the full game are also reachable in the reduced game.
However, the reverse is not guaranteed: paths that do not contain a goal state 
are not necessarily preserved, essentially endowing player~1 with more power.
Consider the (solitaire) reachability game depicted in 
Figure~\ref{fig:ce_Bonneland}, in which all edges belong to player 2 and the 
only goal state is indicated with grey.
Player 2 wins the non-reduced game by avoiding the goal state via the edges 
labelled with $a$ and then $b$.
However, $\{b\}$ is a stubborn set---according to the conditions 
of~\cite{Bonneland2019}---in the initial state, and the dashed transitions are 
thus eliminated in the reduced game.
Hence, player 2 is forced to move the token to the goal state and player 1 wins 
in the reduced game.
In the mean time, the authors of~\cite{Bonneland2019} confirmed and resolved 
the issue in~\cite{bonneland_stubborn_2021}.

The current work is not the first to point out mistakes in POR theory.
In~\cite{Siegel2019}, Siegel presents a flaw in an algorithm that combines POR 
with ample sets and on-the-fly model checking~\cite{Peled1996}.
In that setting, POR is applied on the product of an LSTS and a B\"uchi 
automaton.
We briefly sketch the issue here.
Let $q$ be a state of the LSTS and $s$ a state of the B\"uchi automaton.
While investigating a transition $(q,s) \transition{\act} (q',s')$, condition 
\textbf{C3}, which---like condition \textbf{L}---aims to solve the action 
ignoring problem, incorrectly sets $\redf(q,s') = \enabled(q)$ instead of 
$\redf(q,s) = \enabled(q)$.
The issue is repaired by setting $\redf(q,s) = \enabled(q)$, but only for a 
certain subclass of B\"uchi automata.

The setting considered by Laarman and Wijs~\cite{Laarman2014} is similar: they 
discuss how to apply stubborn sets during parallel nested depth-first search in 
the product of an LSTS and a B\"uchi automaton.
Both the correctness argument and the implementation are based 
on~\cite{Laarman2016}, thus -- by the discussion above -- incorrect in theory, 
but correct in practice.

\section{Conclusion}
\label{sec:conclusion}
We discussed the inconsistent labelling problem for preservation of 
stutter-trace equivalence with stubborn sets.
The issue is relatively easy to repair by strengthening condition \textbf{D1}.
For Petri nets, altering the definition of invisibility can also resolve 
inconsistent labelling depending on the type of atomic propositions.
The impact on applications presented in related works seems to be limited: the 
problem is typically mitigated in the implementation, since it is very hard to 
compute \textbf{D1} exactly.
This is also a possible explanation for why the inconsistent labelling problem 
has not been noticed for so many years.

Since this is not the first error found in POR theory~\cite{Siegel2019}, a more 
rigorous approach to proving its correctness, \eg using proof assistants, would 
provide more confidence.

\section*{Acknowledgements}
The authors would like to thank the anonymous reviewers, including those who 
reviewed the conference version, for their helpful comments.
Special thanks go out to the two journal reviewers.
The first reviewer provided many useful suggestions for improvement and noticed 
that condition \textbf{D1'} can be weakened (see the footnote in 
Section~\ref{sec:strengthen_d1}).
The second reviewer took the significant effort to check all proofs, giving us 
more confidence in the publication.

\bibliographystyle{alpha}
\bibliography{ref}

\end{document}